\documentclass[a4paper,UKenglish,cleveref, autoref, thm-restate]{lipics-v2021}

\usepackage{algorithm}
\usepackage[noend]{algpseudocode}
\usepackage{tikz}
\usetikzlibrary{automata, arrows.meta, bbox, calc, positioning, shapes.geometric}
\usepackage{amsmath}
\usepackage{amsthm}
\usepackage{stmaryrd}
\usepackage{booktabs}
\usepackage{cutwin}
\usepackage[calc]{adjustbox}
\usepackage{wrapfig}
\usepackage{pgfplots}
\usepgfplotslibrary{groupplots}
\pgfplotsset{compat=1.18}

\usepackage{xargs}

\usepackage{xspace}

\usepackage{xstring}

\usepackage[cmex10]{mathtools}
\usepackage{amssymb}
\usepackage{amsfonts}
\usepackage{mathrsfs}
\usepackage{latexsym}
\usepackage{textcomp}
\usepackage{pifont}

\newcommand{\argemp}[2]{\if&#1&\else#2\fi}

\newcommand{\argdef}[2]{\if&#1&#2\else#1\fi}

\newcommand{\argint}[3]{\if&#2&\else#1#2#3\fi}

\newcommand{\argext}[3]{\if&#1&#3\else#1\if&#3&\else#2#3\fi\fi}

\newcommandx{\mthfnt}[3][1=, 2=0]{{
	\IfStrEqCase{#1}
	{%
		{}%
		{#3}%
		{Name}%
		{%
			\IfStrEqCase{#2}
			{%
				{0}{\mathcal{#3}}%
				{1}{\mathscr{#3}}%
				{2}{\mathfrak{#3}}%
				{3}{\mathbb{#3}}%
			}
			[\ensuremath{\clubsuit}]%
		}%
		{Set}%
		{%
			\IfStrEqCase{#2}
			{%
				{0}{\mathrm{#3}}%
				{1}{\mathsf{#3}}%
				{2}{\mathbb{#3}}%
				{3}{\mathbf{#3}}%
			}
			[\ensuremath{\clubsuit}]%
		}%
		{Fun}%
		{%
			\IfStrEqCase{#2}
			{%
				{0}{\mathsf{#3}}%
				{1}{\mathrm{#3}}%
			}
			[\ensuremath{\clubsuit}]%
		}%
		{Rel}%
		{%
			\IfStrEqCase{#2}
			{%
				{0}{\mathit{#3}}%
				{1}{\mathtt{#3}}%
			}
			[\ensuremath{\clubsuit}]%
		}%
		{Sym}%
		{%
			\IfStrEqCase{#2}
			{%
				{0}{\mathtt{#3}}%
				{1}{\mathbf{#3}}%
			}
			[\ensuremath{\clubsuit}]%
		}%
		{Elm}%
		{\mathnormal{#3}}
	}
[\ensuremath{\clubsuit}]%
}}

\newcommand{\mthsub}[1]{\argemp{#1}{\ensuremath{_{\mathnormal{#1}}}}}

\newcommand{\mthsup}[1]{\argemp{#1}{\ensuremath{^{\mathnormal{#1}}}}}

\newcommandx{\mth}[5][1=, 2=0, 4=, 5=]{{\ensuremath{\mthfnt[#1][#2]{#3}\mthsub{#4}\mthsup{#5}}}}

\newcommandx{\mtharg}[6][1=, 2=0, 4=, 5=]{{\mth[#1][#2]{#3}[#4][#5]\ensuremath{\argint{(}{#6}{)}}}}

\newcommand{\mthempty}{\mth[][]}

\newcommand{\mthstyname}{0}
\newcommand{\mthname}[1][]{\mth[Name][\argdef{#1}{\mthstyname}]}

\newcommand{\mthstyset}{0}
\newcommand{\mthset}[1][]{\mth[Set][\argdef{#1}{\mthstyset}]}

\newcommand{\mthstyfun}{0}
\newcommand{\mthfun}[1][]{\mth[Fun][\argdef{#1}{\mthstyfun}]}

\newcommand{\mthstyelm}{0}
\newcommand{\mthelm}[1][]{\mth[Elm][\argdef{#1}{\mthstyelm}]}

\newcommand{\tuple}[1]
{\ensuremath{\!\argint{\langle}{#1}{\rangle}}}

\usepackage{array}   %
\newcolumntype{C}{>{$}c<{$}} %

\newcommand{\Real}{\mathbb{R}}
\newcommand{\Rat}{\mathbb{Q}}

\newcommand{\card}[1]{\mthempty{\argint{\vert}{#1}{\vert}}}
\newcommand{\size}[1]{\mthempty{\argint{\vert\vert}{#1}{\vert\vert}}}

\newcommand{\TEMPORAL}[1]{\mbox{\small\boldmath\(\mathbf{#1}\)}}
\newcommand{\ltlnext}{\TEMPORAL{X}}
\newcommand{\sometime}{\TEMPORAL{F}} 
\newcommand{\always}{\TEMPORAL{G}}
\newcommand{\until}{\,\TEMPORAL{U}\,}

\newcommand{\Ag}{\mthset{N}}
\newcommand{\Ac}{\mthset{Ac}}
\newcommand{\AcProf}{\vec{\Ac}}
\newcommand{\St}{\mthset{St}}
\newcommand{\AP}{\mthset{AP}}

\renewcommand{\Game}{\mthname{G}}

\newcommand{\labFun}{\mthfun{lab}}

\newcommand{\trnFun}{\mthfun{tr}}

\newcommand{\StrSet}{\Sigma}
\newcommand{\strElm}{\sigma}
\newcommand{\strpElm}{\mthelm{\vec{\sigma}}}

\newcommand{\winsym}{\mthset{Win}}
\newcommandx{\Win}[3][1=, 2=, 3=]
{\mthset{\winsym#3}[#1][#2]}

\newcommand{\presym}{\mthfun{Pre}}
\newcommandx{\Pre}[3][1=, 2=, 3=]
{\mthset{\presym#3}[#1][#2]}

\newcommand{\eqsym}{\mthfun{Eq}}
\newcommandx{\Eq}[3][1=, 2=, 3=]
{\mthset{\eqsym#3}[#1][#2]}

\newcommandx{\AFW}[5][1=, 2=, 3=, 4=, 5=]
{\txtargname{AFW#5{\small\argint{$[$}{#1}{$]$}}}[#2][#3]{#4}\xspace}

\def\pspace{\mthfun{PSPACE}\xspace}
\def\npspace{\mthfun{NPSPACE}\xspace}

\def\np{\mthfun{NP}\xspace}
\def\conp{\mthfun{coNP}\xspace}

\renewcommand{\phi}{\varphi}

\newcommand{\jact}{\vec{ac}}

\newcommand{\proj}{\mthfun{proj}}

\newcommand{\pay}{\mthfun{pay}}
\newcommand{\wFun}{\mthfun{w}}
\newcommand{\MP}[1][]{%
	\ifthenelse{\equal{#1}{}}{{\small{\sf mp}}}{{\small{\sf mp}}(#1)}%
	\xspace}

\newcommand{\SetZ}{\mathbb{Z}}

\newcommand{\GRone}{\mthfun{GR(1)}\xspace}

\newcommand{\core}{\mthset{Core}}

\def\src{{\sf src}}
\def\trg{{\sf trg}}
\def\IN{{\sf in}}
\def\OUT{{\sf out}}

\newcommand{\SigmaPTwo}{\Sigma^{\mthfun{P}}_2}
\newcommand{\SigmaPThree}{\Sigma^{\mthfun{P}}_3}
\newcommand{\PiPTwo}{\Pi^{\mthfun{P}}_2}
\newcommand{\PiPThree}{\Pi^{\mthfun{P}}_3}
\newcommand{\QSAT}{{\normalfont \textsc{QSAT}\xspace}}
\newcommand{\DNF}{{\normalfont \textsc{DNF}\xspace}}
\newcommand{\CNF}{{\normalfont \textsc{CNF}\xspace}}

\newcommand{\yclash}{$ y $-clashing\xspace}

\newcommand{\val}{\mthfun{val}}
\newcommand{\PO}{\mthset{PO}}

\newcommand{\C}{\mathbb{C}}
\newcommand{\scc}{\mthset{SCC}}
\newcommand{\conv}{\mthfun{conv}}
\newcommand{\poly}{\mthfun{poly}}
\newcommand{\halfspace}{\mthfun{hspace}}

\newcommand{\sinit}{s_{init}}
\newcommand{\sink}{s_{sink}}

\newcommand{\ac}{\mathrm{ac}}

\renewcommand{\mp}{\textsf{mp}}

\newcommand{\strK}{\sigma_{\mathcal{K}}}

\newcommand{\strP}{\strElm_A}

\newcommand{\coal}{\mathcal{K}}
\newcommand{\player}{A}
\newcommand{\ebendev}{$ \exists $-\textsc{Ben-Dev}\xspace}
\newcommand{\dominated}{\textsc{Dominated}\xspace}
\newcommand{\nonemptiness}{\textsc{Non-Emptiness}\xspace}
\newcommand{\ecore}{\textsc{E-Core}\xspace}
\newcommand{\acore}{\textsc{A-Core}\xspace}

\algnewcommand\Input{\item[ ]{\textbf{input:} }}
\algrenewcommand\Return{\item{\textbf{return} }}

\newcommand{\F}{\mathcal{F}}

\newcommand{\polyset}{\mthset{PS}}

\renewcommand{\bar}{\overline}

\usetikzlibrary{decorations.pathreplacing}
\newcommand\irregularcircle[2]{%
	\pgfextra {\pgfmathsetmacro\len{(#1)+rand*(#2)}}
	+(0:\len pt)
	\foreach \a in {10,20,...,350}{
		\pgfextra {\pgfmathsetmacro\len{(#1)+rand*(#2)}}
		-- +(\a:\len pt)
	} -- cycle
}

\newcommand{\config}{\mthfun{cfg}}
\newcommand{\algdominated}{
\begin{algorithm}[t]
	\caption{Algorithm for \textsc{Dominated}}
	\begin{algorithmic}[1]
		\Input $ \Game, s, \vec{x} $
		\State {guess} a coalition $ C \subseteq \Ag $ and a vector $ (x_i')_{i \in C} \in \Rat^{c}$
		\State $ G^C \leftarrow $ \textsc{Sequentialise$ (\Game,C) $}
		\If{$ \forall i \in C, x_i' > x_i $ and $  (x_i')_{i \in C} \in \val(G^C,s) $}
			\State \textbf{return} YES
		\EndIf
		\State \textbf{return} NO
	\end{algorithmic}
	\label{alg:dominated}
\end{algorithm}
}

\newcommand{\algcompute}{
\begin{algorithm}[t]
	\caption{Algorithm for computing payoff}
	\begin{algorithmic}[1]
		\Input $ \Game, \strpElm $
		\State {guess} $ s^j $ and a vector $ (q_1^j,\dots,q_n^j) \in \prod_{i \in \Ag} Q_i $

		\State $ k,l \leftarrow $ \textsc{ComputeIndex$ \left( \Game, \strpElm, (s^j,q_1^j,\dots,q_n^j) \right) $}

			\State \textbf{return} $ \left( \frac{ \sum_{t=k}^{l} \wFun_{i} \left( \pi(\strpElm)[t] \right)}{l-k} \right)_{i \in \Ag} $

	\end{algorithmic}
	\label{alg:compute}
\end{algorithm}
}

\newcommand{\algnonemptiness}{
	\begin{algorithm}[H]
		\caption{Algorithm for \textsc{Non-Emptiness}}
		\begin{algorithmic}[1]
			\Input $ \Game $
			\State $ G^{\Ag} \leftarrow $ \textsc{Sequentialise$ (\Game,\Ag) $}
			\State {guess} a vector $ \vec{x} \in \Rat^{n}$ s.t. $ \vec{x} \in \val(G^{\Ag}) $
			
			\If{ $ (\Game,\sinit,\vec{x}) \in $ \protect\textsc{Dominated} (Alg.~\ref{alg:dominated})}
			\State \textbf{return} NO
			\EndIf
			\State \textbf{return} YES
		\end{algorithmic}
		\label{alg:nonemptiness}
	\end{algorithm}
}

\newcommand{\algecore}{
	\begin{algorithm}[t]
		\caption{Algorithm for \textsc{E-Core}}
		\begin{algorithmic}[1]
			\Input $ \Game, \phi $
			\State $ G^{\Ag} \leftarrow $ \textsc{Sequentialise$ (\Game,\Ag) $}
			\State {guess} a vector $ \vec{x} \in \Rat^{n}$ s.t. $ \vec{x} \in \val(G^{\Ag}) $ and set of states $ S \subseteq \St $
			
			\If{ there is no $ s \in S $ s.t. $ (\Game,s,\vec{x}) \in $ \protect\textsc{Dominated} (Algorithm~\ref{alg:dominated})}
			\State $ \Game{[S]} \leftarrow $ \textsc{UpdateArena}$ (\Game,S) $
				\If{ $ \pi \models \psi $ for some $ \pi \in \Game{[S]} $}
					\State \textbf{return} YES
				\EndIf
			\EndIf
			\State \textbf{return} NO	
		\end{algorithmic}
		\label{alg:ecore}
	\end{algorithm}
}

\title{Characterising and Verifying the Core in Concurrent Multi-Player Mean-Payoff Games\\ (Full Version)} %

\titlerunning{Characterising and Verifying the Core in Concurrent Multi-Player Mean-Payoff Games} %

\author{Julian Gutierrez}{Monash University, Australia }{julian.gutierrez@monash.edu}{}{}%

\author{Anthony W. Lin}{TU Kaiserslautern, Germany }{lin@cs.uni-kl.de}{}{}

\author{Muhammad Najib}{Heriot-Watt University, UK }{m.najib@hw.ac.uk}{https://orcid.org/0000-0002-6289-5124}{}

\author{Thomas Steeples}{University of Oxford, UK }{thomas.steeples@cs.ox.ac.uk}{}{}

\author{Michael Wooldridge}{University of Oxford, UK }{mjw@cs.ox.ac.uk}{}{}

\authorrunning{J. Gutierrez, A.W. Lin, M. Najib, T. Steeples, and M. Wooldridge} %

\Copyright{Julian Gutierrez, Anthony W. Lin, Muhammad Najib, Thomas Steeples, Michael Wooldridge} %

\ccsdesc[500]{Theory of computation~Logic and verification}
\ccsdesc[500]{Theory of computation~Verification by model checking}
\ccsdesc[500]{Theory of computation~Solution concepts in game theory}

\keywords{Concurrent games, multi-agent systems, temporal logic, game theory} %

\category{} %

\relatedversion{} %

\acknowledgements{We wish to thank anonymous reviewers for their useful feedback. }%

\nolinenumbers

\begin{document}

\maketitle

\begin{abstract}
Concurrent multi-player mean-payoff games are important models for systems of agents with individual, non-dichotomous preferences. 
Whilst these games have been extensively studied in terms of their equilibria in non-cooperative settings, this paper explores an alternative solution concept: the core from cooperative game theory. This concept is particularly relevant for cooperative AI systems, as it enables the modelling of cooperation among agents, even when their goals are not fully aligned.
Our contribution is twofold.
First, we provide a characterisation of the core using discrete geometry techniques and establish a necessary and sufficient condition for its non-emptiness. We then use the characterisation to prove the existence of polynomial witnesses in the core. Second, we use the existence of such witnesses to solve key decision problems in rational verification and provide tight complexity bounds for the problem of checking whether some/every equilibrium in a game satisfies a given LTL or \GRone specification. Our approach is general and can be adapted to handle other specifications expressed in various fragments of LTL without incurring additional computational costs.
\end{abstract}

\section{Introduction}\label{secn:introduction}
\textit{Concurrent games}, where agents interact over an infinite sequence of rounds by choosing actions simultaneously, are one of the most important tools for modelling multi-agent systems. This model has received considerable attention in the research community (see, e.g.,~\cite{Alur2002,bouyer2015pure,GutierrezHW17,fisman2010rational,kupferman2016synthesis}).
In these games, the system evolves based on the agents’ choices, and their preferences are typically captured by associating them with a \textit{Boolean objective} (e.g., a \textit{temporal logic formula}) representing their goal. Strategic issues arise as players seek to satisfy their own goals while taking into account the goals and rational behaviour of other players. 
Note that the preferences induced by such goals are \textit{dichotomous}: a player will either be satisfied or unsatisfied. 
However, many systems require richer models of preferences that capture issues such as resource consumption, cost, or system performance~\cite{CdAHS03,CCHKM05,4484790}. 

\textit{Mean-payoff games}~\cite{EM79,zwick1996the} are widely used to model the quantitative aspects of systems.
Whilst much research has been conducted on \textit{non-cooperative} mean-payoff games and solution concepts such as \textit{Nash equilibrium} (NE) and \textit{subgame perfect equilibrium} (e.g., \cite{ummels2011the,BriceRB21,BdPS13}), this paper focuses on a \textit{cooperative} setting. In this setting, players can reach binding agreements and form coalitions to collectively achieve better payoffs or eliminate undesirable outcomes\footnote{We emphasise that this paper concerns the \textit{outcome} of games \textit{when} such agreements can be reached. The mechanism for agreements is assumed to be exogenous and beyond the scope of this paper.}. As a result, NE and its variants may not be suitable for examining the stable behaviours that arise in these types of games. For example, in the Prisoner's Dilemma game, players can avoid mutual defection, which is the unique NE, by establishing binding agreements~\cite{ray1997equilibrium}. Thus, analysing games through the lens of cooperative game theory poses distinct challenges and is important in and of itself.
This paradigm is particularly relevant for modelling and analysing \textit{cooperative AI} systems, which have recently emerged as a prominent topic~\cite{dafoe2020open,dafoe2021cooperative,conitzer2023foundations,elisabertino2020artificial}. In these systems, agents are able to communicate and benefit from cooperation, even when their goals are not fully aligned. 
We illustrate that this is also the case in the context of mean-payoff games in Example~\ref{example}.

 We focus on a solution concept from cooperative game theory known as the \textit{core}~\cite{aumann1961core,scarf1967core,Gutierrez2019}, which is the most widely-studied solution concept for cooperative games. Particularly, we study the core of mean-payoff games where players have access to \textit{finite but unbounded memory strategies.} The motivation is clear, as finite-memory strategies are sufficiently powerful for implementing LTL objectives while being realisable in practice. Our main contribution is twofold: First, we provide a characterisation of the core using techniques from discrete geometry\footnote{We note that linear programming and convex analysis are well-established tools for studying the core in traditional economics (see e.g.,~\cite{gillies1959solutions,bondareva1963some,scarf1967core,uyanik2015nonemptiness}). However, the settings and contexts (e.g., the game models) of these previous works differ from those of the present work, and their results do not automatically carry over.} (cf. logical characterisation in~\cite{Gutierrez2019,gutierrez2023cooperative}) and establish a necessary and sufficient condition for its non-emptiness. We believe that our characterisation holds value in its own right, as it connects to established techniques used in game theory and economics. This has the potential to enable the application of more sophisticated methods and computational tools (e.g., linear programming solvers) in the area of \textit{rational verification}~\cite{abate2021rational,GutierrezHW17}.
 Second, we provide tight complexity bounds for key decision problems in {rational verification} with LTL and \GRone~\cite{Bloem2012} specifications (see Figure \ref{fig:summary_of_results}). \GRone is a LTL fragment that has been used in various domains~\cite{9812068,camacho2019towards,filippidis2016control,MS12} and covers a wide class of common LTL specification patterns~\cite{maoz2015gr}. Our approach to solving rational verification problems is very general and can be easily adapted for different LTL fragments beyond \GRone. This is the first work to study the core of mean-payoff games with \textit{finite but unbounded memory strategies}, and to explore the complexity of problems related to the rational verification of such games in this setting.
 
 \subparagraph*{Related Work.}
The game-theoretical analysis of temporal logic properties in multi-agent systems has been studied for over a decade (see e.g., \cite{fisman2010rational,chatterjee2010strategy,mogavero2014reasoning,kupferman2016synthesis,gutierrez2015iterated,wooldridge2016rational}). However, most of the work has focused on a non-cooperative setting. Recently, there has been an increased interest in the analysis of concurrent games in a cooperative setting.
The core has been studied in the context of deterministic games with dichotomous preferences by \cite{Gutierrez2019,gutierrez2023cooperative} using the logics ATL*~\cite{Alur2002} and SL~\cite{mogavero2014reasoning}. However, as far as we are aware, there are no extensions of these logics that adopt mean-payoff semantics. Quantitative extensions exist~\cite{bulling2022combining,bouyer2023reasoning}, and the core is studied in~\cite{bouyer2023reasoning} using the logic SL[$ \mathcal{F} $] that extends SL with quantitative satisfaction value, but the semantics of these logics are not defined on mean-payoff conditions and thus cannot be used to reason about the core of mean-payoff games. In the stochastic setting, \cite{HLNW21} examines the core in stochastic games with LTL objectives under the \textit{almost-sure} satisfaction condition. The approach relies on \textit{qualitative parity logic}~\cite{berthon2020mixing} and is not applicable to mean-payoff objectives.
Closer to our work is \cite{steeples2021mean} which studies the core of multi-player  mean-payoff games with Emerson-Lei condition~\cite{EL87} in the \textit{memoryless} setting. Whilst memoryless strategies are easy to implement, finite-memory and arbitrary mathematical strategies offer greater richness. For instance, players can achieve higher payoffs and implement LTL properties with finite-memory strategies, which may not be possible with memoryless ones (see Example~\ref{example}). The approach proposed in \cite{steeples2021mean}, which involves using a non-deterministic Turing machine to guess the correct strategies, is not applicable in the present work's setting. This is because players may have finite but unbounded memory strategies, and as such, strategies may be arbitrarily large. To address this limitation, we propose a new approach that can handle such scenarios.

\subparagraph*{Organisation.}
The rest of the paper is structured as follows. Section~\ref{secn:preliminaries} provides an overview of temporal logics, multi-player mean-payoff games, relevant game-theoretic concepts, and key mathematical concepts. Section~\ref{secn:characterisations} develops a method to characterise the core using discrete geometry techniques, leading to a crucial result for Section~\ref{secn:decision_problems}, where we determine the complexity of several decision problems. Finally, Section~\ref{secn:concluding_remarks} offers concluding remarks.

\section{Preliminaries}\label{secn:preliminaries}
Given any set $ X $, we use $ X^*, X^{\omega} $ and $ X^+ $ for, respectively, the sets of finite, infinite, and non-empty finite sequences of elements in $ X $. For $ Y \subseteq X $, we write $ X_{-Y} $ for $ X \setminus Y $ and $ X_{-i} $ if $ Y = \{i\} $. We extend this notation to tuples $ w = (x_1,...,x_k,...,x_n) \in X_1 \times \cdots \times X_n $, and write $ w_{-k} $ for $ (x_1,...,x_{k-1},x_{k+1},...,x_n) $. Similarly, for sets of elements, we write $ w_{-Y} $ to denote $ w $ without each $ x_k$, for $ k \in Y $. For a sequence $ v $, we write $ v[t] $ or $ v^t $ for the element in position $ t + 1 $ in the sequence; for example, $ v[0] = v^0 $ is the first element of $ v $.
\subparagraph*{Mean-Payoff.} 
For an infinite sequence of real numbers, \(r^0 r^1 r^2 \cdots \in \mathbb{R}^\omega\), we define the \emph{mean-payoff} value of \(r\), denoted \(\MP(r)\), to be the quantity,
$
	\MP(r) = \liminf_{n \to \infty} \frac{1}{n} \sum_{i=0}^{n-1} r^i.
$

\subparagraph*{Temporal Logics.} We use
LTL~\cite{Pnueli1977} with the usual temporal operators, \(\ltlnext\) (``next'') and \(\until\) (``until''), and the derived operators $ \always $ (``always'') and $ \sometime $ (``eventually''). We also use \(\GRone\)~\cite{Bloem2012}, a fragment of LTL given by formulae written in the following form: \[
(\always \sometime \psi_1 \wedge \cdots \wedge \always \sometime \psi_m) \to (\always \sometime \phi_1 \wedge \cdots \wedge \always \sometime \phi_n)
\text{,}
\] where each subformula \(\psi_i\) and \(\phi_i\) is a Boolean combination of atomic propositions.
Additionally, we also utilise an extension of LTL known as LTL$ ^{lim\Sigma} $~\cite{Boker2014} that allows mean-payoff assertion such as $ \mp(v) \geq c $ for a numeric variable $ v $ and a constant number $ c $, which asserts that the mean-payoff of $ v $ is greater than or equal to $ c $ along an entire path.
The satisfaction of temporal logic formulae is defined using standard semantics. We use the notation $ \alpha \models \phi $ to indicate that the formula $ \phi $ is satisfied by the infinite sequence $ \alpha $.

\subparagraph*{Arenas.}
An \emph{arena} is a tuple \(A = \tuple{\Ag,  \{\Ac_i\}_{i \in \Ag}, \St, \sinit, \trnFun, \labFun}\)  where \(\Ag\), \(\Ac_i\), and \(\St\) are finite non-empty sets of \emph{players}, \emph{actions} for player $ i $, and \emph{states}, respectively; \(\sinit \in \St\) is the \emph{initial state}; \(\trnFun : \St \times \AcProf \rightarrow \St\) is a \emph{transition function} mapping each pair consisting of a state \(s \in \St\) and an \emph{action profile} \(\jact \in \AcProf = \Ac_1 \times \cdots \times \Ac_n \), with one action for each player, to a successor state; and \(\labFun: \St \to 2^{\AP}\) is a labelling function, mapping every state to a subset of \emph{atomic propositions}.

A \emph{run} \(\rho = (s^0, \jact^0), (s^1, \jact^1) \cdots\) is an infinite sequence in \((\St \times \AcProf)^{\omega}\) such that \(\trnFun(s^k, \jact^k) = s^{k + 1}\) for all \(k\). 
Runs are generated in the arena by each player~\(i\) selecting a \emph{strategy} \(\strElm_i\) that will define how to make choices over time. 
A strategy for \( i \) can be understood abstractly as a function \( \strElm_{i}: \St^+ \to \Ac_i \) which maps sequences (or histories) of states into a chosen action for player $ i $. A strategy $ \strElm_{i} $ is a \textit{finite-memory} strategy if it can be represented by a finite state machine \(\strElm_{i} = (Q_{i}, q_{i}^{0}, \delta_i, \tau_i) \), where \(Q_{i}\) is a finite and non-empty set of \emph{internal states}, \( q_{i}^{0} \) is the \emph{initial state}, \(\delta_i: Q_{i} \times \St \rightarrow Q_{i} \) is a deterministic \emph{internal transition function}, and \(\tau_i: Q_{i} \rightarrow \Ac_i\) an \emph{action function}. A \textit{memoryless} strategy $ \strElm_{i}: \St \to \Ac_i $ chooses an action based only on the current state of the environment.
We write \(\StrSet_i\) for the set of strategies for player \(i\).

A \emph{strategy profile} \(\strpElm = (\strElm_1, \dots, \strElm_n)\) is a vector of strategies, one for each player. 
Once a state \(s\) and profile \(\strpElm\) are fixed, the game has an \emph{outcome}, i.e., a path in \(A\), denoted by \(\pi(\strpElm, s)\). In this paper, we assume that players' strategies are \textit{finite-memory} and \textit{deterministic}, as such, \(\pi(\strpElm, s)\) is the \textit{unique} path induced by \(\strpElm\), that is, the sequence \(s^0s^1s^2\ldots\) such that \( s^0 = s \), \(s^{k + 1} = \trnFun (s^k, (\tau_1(q^k_1), \ldots, \tau_n(q^k_n)))\), and
\(q^{k + 1}_i = \delta_i(q^k_i, s^k)\), for all \(k \geq 0\). Note that such a path is \textit{ultimately periodic} (i.e., a lasso).
We simply write \( \pi(\strpElm) \) for \( \pi(\strpElm,\sinit) \). We extend this to run induced by $ \strpElm $ in a similar way, i.e., $ \rho(\strpElm) = (s^0,\jact^0), (s^1,\jact^1), \dots $. For an element of a run $ \rho(\strpElm)[k] = (s^k,\jact^k) $, we associate the configuration $ \config(\strpElm,k) = (s^k,q_1^k,\dots,q_n^k) $ with $ \tau_i(q_i^k) = \ac_i^k $ for each $ i $.

\subparagraph*{Multi-Player Games.}
A \emph{multi-player game} is obtained from an arena \(A\) by associating each player with a goal. We consider multi-player games with mean-payoff  goals. A \textit{multi-player mean-payoff game} (or simply a \textit{game}) is a tuple \(\Game = \tuple{A, (\wFun_{i})_{i \in \Ag}}\), where \(A\) is an arena and \(\wFun_{i}: \St \to \SetZ\) is a function mapping, for every player~\(i\), every state of the arena into an integer number. Given a game \(\Game = \tuple{A, (\wFun_{i})_{i \in \Ag}}\) and a strategy profile $ \strpElm $, an outcome \(\pi(\strpElm)\) in \(A\) induces a sequence \(\labFun(\pi(\strpElm)) = \labFun(s^0) \labFun(s^1) \cdots\) of sets of atomic propositions, and for each player~\(i\), the sequence \(\wFun_i(\pi(\strpElm)) = \wFun_i(s^0) \wFun_i(s^1) \cdots\) of weights. The \textit{payoff} of player~\(i\) is \(\pay_i(\pi(\strpElm)) = \MP(\wFun_{i}(\pi(\strpElm)))\). By a slight abuse of notation, we write \( \pay_i(\strpElm) \) for \( \pay_i(\pi(\strpElm)) \), and $ \pi(\strpElm) \models \phi $ or $ \strpElm \models \phi $ for $ \labFun(\pi(\strpElm)) \models \phi $ for some temporal logic formula $ \phi $.

\subparagraph*{Solution Concept.} We focus on a solution concept known as the 
\emph{core}~\cite{aumann1961core,scarf1967core,Gutierrez2019}. To understand the concept of the core, it might be helpful to compare it with the NE and how each can be characterised by \textit{deviations}. Informally, a NE is a strategy profile from which no player has any incentive to \emph{unilaterally} deviate. On the other hand, the core comprises strategy 
profiles from which no \emph{coalitions} of agents can deviate such that \textit{every} agent in the coalition is strictly better off, regardless of the actions of the remaining players.

Formally, we say that a strategy profile \(\strpElm\) is in the core if for all coalitions \(C \subseteq \Ag\), and strategy profiles \( \strpElm_{C}' \), there is some counter-strategy profile \( \strpElm_{-C}' \) such that \( \pay_i(\strpElm) \geq \pay_i(\strpElm_{C}',\strpElm_{-C}') \), for some \(i \in C\). Alternatively, as we already discussed above, we can characterise the core by using the notion of \emph{beneficial deviations}: Given a strategy profile \(\strpElm\) and a coalition \(C \subseteq \Ag, C \neq \varnothing \), we say that the strategy profile \(\strpElm_C^\prime\) is a \emph{beneficial deviation} if for all counter-strategies \(\strpElm_{-C}'\), we have \(\pay_i((\strpElm_C^\prime, \strpElm_{-C}^\prime)) > \pay_i(\strpElm)\) for all \(i \in C\). The core then consists of those strategy profiles which admit no beneficial deviations; note that these two definitions are equivalent. For a given game \( \Game \), let \( \core(\Game) \) denote the set of strategy profiles in the core of \( \Game \).

\begin{figure}[ht]
	\centering
	\scalebox{0.7}{
\begin{tikzpicture}[state/.style={circle, draw, minimum size=1cm}, node distance=1.5cm]
	\node[state, label=below:{$ (0,0) $}] (s^0) { $m$};
	\node[] (start) [above =0.5cm of s^0] {};
	\node[state, label=below:{$ (0,1) $}] (s^2) [right = of s^0] { $r$};
	\node[state, label=below:{$ (1,0) $}] (s^1) [left = of s^0] { $l$};
	
	\draw [-{Latex[width=2mm]}]
	(s^0) edge[bend right] node[above]{\small $(L,L)$} (s^1)
	(s^1) edge[bend right] node[below]{\small $(\ast,R)$} (s^0)
	
	(s^0) edge[bend right] node[below]{\small $(R,R)$} (s^2)
	(s^2) edge[bend right] node[above]{\small $(L,\ast)$} (s^0)
	(s^0) edge[loop, out=70, in=110, distance=1cm] node[above, align=center]{\small $(L,R)$\\ \small $(R,L)$} (s^0)
	(s^1) edge[loop, out=200, in=160, distance=1cm] node[left]{\small $(L,L)$} (s^1)
	(s^2) edge[loop, out=20, in=340, distance=1cm] node[right]{\small $(R,R)$} (s^2)
	;
\end{tikzpicture}
}
	\caption{Arena for Example~\ref{example}. The symbol $ \ast $ is a wildcard that matches all possible actions. 
		\label{fig:example}}
	\vspace{-0.3cm}
\end{figure}
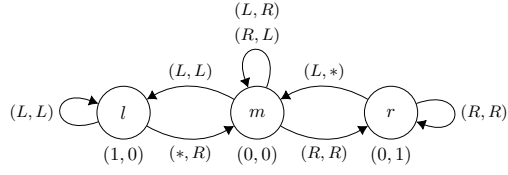

\begin{example}
	\label{example}
	We illustrate how the core differs from NE, and how cooperation and memory affect the outcome of a game.
	Consider a game consisting of two players $ \{1,2\} $. The arena is depicted in Figure~\ref{fig:example}, and the players are initially in $ m $. Each player has two actions: $ L $ and $ R $. Player $ 1 $ (resp. $ 2 $) gets $ 1 $ when the play visits $ l $ (resp. $ r $)---e.g., tasks assigned to the players, for which they are rewarded upon completion. However, these states can only be visited by agreeing on the actions (e.g., tasks that must be carried out by multiple robots). Observe that player $ 1 $ (resp. $ 2 $) always choosing $ L $ (resp. $ R $) is a NE, and a ``bad'' one since each player receives a payoff of $ 0 $. On the other hand, this bad equilibrium is not included in the core: the players can coordinate/cooperate to alternately visit $ l $ and $ r $ and obtain higher payoffs (i.e., each receives $ \frac{1}{4} $). Furthermore, observe that to execute this plan, the players must remember previously visited states (i.e., finite-memory strategies are necessary). This outcome also corresponds to the liveness property $ \always \sometime l \wedge \always \sometime r $ (``the tasks will be completed infinitely often''), which cannot be realised using memoryless strategies.
\end{example}

\subparagraph*{Vectors and Inequations.} 
Given two vectors $ \vec{a}, \vec{b} \in \Rat^d $ the notation $ \vec{a} \geq \vec{b} $ corresponds to the \textit{component-wise inequality}, and let $ \size{\vec{a}} = d + \sum_{i \in \llbracket 1,d \rrbracket} \size{a_i} $, $ a_i $ is represented using the usual binary encoding of numerators/denominators. The \emph{linear function} \( f_{\vec{a}} : \Real^d \to \Real \) is the function \( f_{\vec{a}}(\vec{x}) = \sum_{i \in \llbracket 1,d \rrbracket} a_i \cdot x_i \). A \emph{linear inequation} is a pair \( (\vec{a},b) \) where \( \vec{a} \in \Rat^d \setminus \{\vec{0}\} \) and \( b \in \Rat \). The size of $ (\vec{a},b) $ is $ \size{(\vec{a},b)} = \size{\vec{a}} + \size{b} $. The \emph{half-space} corresponding to \( (\vec{a},b) \) is the set \( \halfspace(\vec{a},b) = \{ \vec{x} \in \Real^d \mid f_{\vec{a}}(\vec{x}) \leq b \} \). A \emph{linear inequality system} is a set \( \lambda = \{ (\vec{a}_1,b_1),\dots,(\vec{a}_l,b_l) \} \) of linear inequations. A \emph{polyhedron} generated by \( \lambda \) is denoted by \( \poly(\lambda) = \bigcap_{(\vec{a},b) \in \lambda} \halfspace(\vec{a},b) \). Let \( P \) be a polyhedron in \( \Real^{d} \) and \( C \subseteq D = \{1,\dots,d\} \), and let $ c = \card{C} $. The \emph{projection} of \( P \subseteq \Real^d \) on variables with indices in \( C \) is the set \( \proj_C(P) = \{ \vec{x} \in \Real^{c} \mid \exists \vec{y} \in P \wedge \forall i \in C, y_i = x_i \} \).

\section{Characterising the core}\label{secn:characterisations}
In this section, we provide a characterisation of the core and other important concepts which we will use to prove our complexity results.

\subparagraph*{Multi-Mean-Payoff Games.}
\emph{Multi-mean-payoff games} (MMPGs)~\cite{Chatterjee2010,Velner2015} are similar to two-player, turn-based, zero-sum mean-payoff games, except the states of the game graph are labelled with \( k \)-dimensional integer vectors representing the weights. 
Player 1's objective is to maximise the mean-payoff of the $ k $-dimensional weight function. Note that since the weights are multidimensional, there is not a unique maximal value in general.

Formally, a multi-mean-payoff game \(G\) is a tuple, \(G = (V_1, V_2, E, w)\), where \(V_1, V_2\) are the \emph{states} controlled by player 1 and 2 respectively, with \(V \coloneqq V_1 \cup V_2\) and \(V_1 \cap V_2 = \varnothing\); \(E \subseteq V \times V\) is a set of \emph{edges}; \(w : V \to \mathbb{Z}^k\) is a \emph{weight function} with $ k \in \mathbb{N} $.
Given a start state \(v^0 \in V_i\), player \(i\) chooses an edge \((v^0, v^1) \in E\), and the game moves to state \(v^1 \in V_j\). Then player \(j\) chooses an edge and the game moves to the specified state, and this continues forever. Paths are defined in the usual way and for a path $ \pi $, the payoff $ \pay(\pi) $ is the vector $ (\MP(w_1(\pi)),\dots,\MP(w_k(\pi))) $.
It is shown in~\cite{Velner2015} that memoryless strategies suffice for player 2 to act optimally, and that the decision problem which asks if player 1 has a strategy that ensures $ \pay(\pi) \geq \vec{x} $ from a given state and for some $ \vec{x} \in \Real^k $ is \conp-complete.

We consider a \emph{sequentialisation} of a game where players are partitioned into two coalitions, \( C \subseteq \Ag \) and \( -C = \Ag \setminus C \).  This game is modelled by a MMPG where coalition \( C \) acts as player 1 and \( -C \) as player 2. The \( k \)-dimensional vectors represent the weight functions of players in \( C \). In the case \( C = \Ag \), player 2 is a ``dummy'' player with no influence in the game. 

\begin{definition}
    Let \( \Game = (A,(\wFun_i)_{i \in \Ag}) \) be a game with \( A = (\Ag,\Ac,\St,\sinit,\trnFun,\labFun) \) and let \( C \subseteq \Ag \). The sequentialisation of \( \Game \) with respect to \( C \) is the (turn-based two-player) MMPG \( G^C = (V_1,V_2,E, w) \) where \( V_1 = St \), \( V_2 = \St \times \AcProf_C \); \( w: V_1 \cup V_2 \to \mathbb{Z}^{c} \) is such that \( w_i(s) = w_i(s,\jact_{C}) = \wFun_{i}(s) \); and \( E = \{ (s,(s,\jact_{C})) \in \St \times (\St \times \AcProf_C) \} \cup \{ ((s,\jact_{C}),s') \in (\St \times \AcProf_C) \times \St : \exists \jact_{-C} \in \AcProf_{-C}. s' = \trnFun(s,(\jact_{C},\jact_{-C})) \} \).\footnote{For \( C = \Ag \), the set \( \AcProf_{-C} \) is empty, and the transition is fully characterised by \( \AcProf_C \). We keep the current notation to avoid clutters.}
    
\end{definition}
\noindent The construction above is clearly polynomial in the size of the original game $ \Game $.

Let \( \StrSet^M_2 \) be the set of memoryless strategies\footnote{Here we define a strategy as a mapping from sequences of states to a successor state $ \sigma_i : V^\ast V_i \to V $ for $ i \in \{1,2\} $. A strategy is memoryless when it chooses a successor based on the current state $ \sigma_i : V_i \to V $.} for player 2. For a strategy \( \strElm_2 \in \StrSet^M_2 \), the game induced by applying such strategy is given by \( G^C[\strElm_2] = (V_1,V_2,E',(w_i)_{i \in C}) \) where \( E' = \{(s,s') \in E \mid s \in V_1 \vee (s \in V_2 \wedge \strElm_2(s) = s') \} \). That is, a subgame in which player 2 plays according to the memoryless strategy \( \strElm_2 \).

\subparagraph*{Enforceable Values and Pareto Optimality.} We present the definitions of \textit{enforceable values} and \emph{Pareto optimality} in MMPGs~\cite{Brenguier2015} below, which we will use for our characterisation of the core. 

\begin{definition}
	For a MMPG \( G^C \) and a state \( s \in V_1 \cup V_2 \), define the set of enforceable values that player 1 can guarantee from state $ s $ as:
	\[
		\val(G^C,s) = \{ \vec{x} \in \mathbb{R}^{c	} \mid \exists \strElm_1 \forall \strElm_2 \forall j \in C :  
		 x_j \leq \MP_j(w_j(\pi((\strElm_1,\strElm_2),s))) \}.
	\]
\end{definition}

A vector \( \vec{x} \in \mathbb{R}^{c} \) is $ C$\emph{-Pareto optimal from \( s \)} (or simply \textit{Pareto optimal} when $ C = \Ag $ or $ C $ is clear from the context, and $ s = \sinit $) if it is maximal in the set \( \val(G^C,s) \). The set of Pareto optimal values is called \emph{Pareto set}, formally defined as:
\[
\PO(G^C,s) = \{ \vec{x} \in \val(G^C,s) \mid \neg \exists \vec{x}' \in \val(G^C,s) :
\vec{x}' \geq \vec{x} \land \exists i \,\text{s.t.}\, x_i^\prime > x_i \}.
\]
When \( s = \sinit \), we simply write \( \val(G^C) \) and \( \PO(G^C) \). We naturally extend Pareto optimality to strategy profiles: a strategy profile $ \strpElm $ is $ C $-{Pareto optimal} if $ (\pay_i(\strpElm))_{i \in C} \in \PO(G^C) $. 

A notable aspect of the core in mean-payoff games is that it generally does not coincide with Pareto optimality, as shown in Propositions~\ref{prpn:core-not-pareto-optimal} and \ref{prpn:core-pareto-optimal} below (the proofs are provided in the technical appendix). This stands in sharp contrast to conventional cooperative (transferable utility, superadditive) games in which the core is always included in the Pareto set~\cite[pp.~24--25]{chalkiadakis}.

\begin{restatable}{proposition}{PROPstrict}
\label{prpn:core-not-pareto-optimal}
    There exist games \(\Game\) such that \(\vec{\sigma} \in \core(\Game)\) and \(\vec{\sigma}\) is not Pareto optimal.
\end{restatable}

\begin{restatable}{proposition}{PROPPO}
\label{prpn:core-pareto-optimal}
	There exist games \(\Game\) such that \(\vec{\sigma}\) is Pareto optimal and \(\vec{\sigma} \not\in \core(\Game)\).
\end{restatable}

\subparagraph*{Discrete Geometry and Values.} 
To characterise the core, we utilise techniques from discrete geometry. First, we provide the definitions of two concepts: \textit{convex hull} and \textit{downward closure}. The \textit{convex hull} of a set $ X \subseteq \Real^d $ is the set $ \conv(X) = \{ \sum_{\vec{x} \in X} a_{\vec{x}} \cdot \vec{x} \mid \forall \vec{x} \in X, a_{\vec{x}} \in [0,1] \wedge \sum_{\vec{x} \in X} a_{\vec{x}} = 1 \} $. The \textit{downward closure} of a set $ X \subseteq \Real^d $ is the set $ \downarrow\! X = \{ \vec{x} \in \Rat^d \mid \exists \vec{x}' \in X, \forall i \in \llbracket 1,d \rrbracket, x_i \leq x_i' \} $. Note that if the set $ X $ is finite, then $ \conv(X) $ and $ \downarrow\conv(X) $ are convex polyhedra, thus can be represented by intersections of some finite number of half-spaces~\cite[Theorem~3.1.1]{grunbaum2003convex}.

Now, observe that the downward closure of the Pareto set is equal to the set of values that player 1 can enforce, that is, \( \downarrow \PO(G^C,s) = \val(G^C,s) \).
The set \( \val(G^C) \) can also be characterised by the set of \textit{simple cycles} and \textit{strongly connected components} (SCCs) in the arena of \( G^C \)~\cite{ADMW09}. A \emph{simple cycle} within \( S \subseteq (V_1 \cup V_2) \) is a finite sequence of states \( o = s^0s^1\cdots s^k \in S^* \) with $ s^0 = s^k $ and for all $ i $ and $ j, 0 \leq i < j < k, s^i \neq s^j $. 
Let \( \C(S) \) be the set of simple cycles in \( S \), and \( \scc(G^C{[\strElm_2]}) \) the set of SCCs reachable from \( \sinit \) in \( G^C{[\strElm_2]} \). The set of values that player 1 can enforce is characterised by the intersection of all sets of values that it can achieve against memoryless strategies of player 2. Formally, we have the following~\cite[Theorem 4]{Brenguier2015}:
\[
	\val(G^C) =  \bigcap_{\strElm_2 \in \StrSet^M_2} \bigcup_{S \in \scc(G^C[\strElm_2])} 
	\downarrow \conv \left( \left\{ \left(\frac{ \sum_{j=0}^{k} w_i(o^j)}{\card{o}}\right)_{i \in C} \,\middle\vert\, o \in \C(S) \right\} \right). 
\]

With these definitions in place, we first obtain the following lemma, which shows that the set of enforceable values has polynomial representation.

\begin{lemma}\label{lem:representation}
	The set $ \val(G^C) $ can be represented by a finite union of polyhedra $ P^C_1,\dots,P^C_k $, each of them defineable by a system of linear inequations $ \lambda^C_j $. Moreover, each linear inequation $ (\vec{a},b) \in \lambda^C_j $ can be represented polynomially in the size of $ G^C $.
\end{lemma}

\begin{proof}
 Let $ X = \{ \vec{x}_1,\dots,\vec{x}_m \} $ be the set of extreme points of $ \conv ( \{ (\frac{ \sum_{j=0}^{k} w_i(o_j)}{\card{o}})_{i \in C} \mid o \in \C(S) \} )  $ for a given $ S \in \scc(G^C[\strElm_2]) $. Observe that $ X $ corresponds to the set of simple cycles in $ S $, as such, for each $ \vec{x} \in X $ we have $ \size{\vec{x}} $ that is of polynomial in the size of $ G^C $. As shown in \cite[Theorem 3]{Brenguier2015},  $ \downarrow \conv(X) $ has a system of inequations $ \lambda $ whose each inequation has representation polynomial in $ c $ and $ \log_2(\max \{\size{\vec{x}} \mid \vec{x} \in X \}) $. Since this holds for each $ \strElm_2 \in \StrSet^M_2 $ and for each SCC in $ G^C[\strElm_2] $, we obtain the lemma.
\end{proof}

Let $ \polyset(G^C) $ denote the set of polyhedra whose union represents $ \val(G^C) $, and for a polyhedron $ P^C_j \in \polyset(G^C) $, we denote by $ \mathcal{H}^C_j $ the set of half-spaces whose intersection corresponds to $ P^C_j $.

\subparagraph*{Polynomial Witness in the Core.}
A \textit{polynomial witness} in the core of $ \Game $ is a vector $ \vec{x} \in \Rat^n $ such that there exists $ \strpElm \in \core(\Game) $ where $ (\pay_i(\strpElm))_{i \in \Ag} = \vec{x} $ and $ \vec{x} $ has a polynomial representation with respect to $ \Game $.
The rest of this section focuses on characterising the core (Theorem~\ref{thm:characterisation}) and showing the existence of a polynomial witness in a non-empty core (Theorem~\ref{thm:polywitness}). We start by introducing some concepts and proving a couple of lemmas.

\begin{definition}
	Given a set of player $ \Ag $ and a coalition $ C \subseteq \Ag $. The \emph{inclusion mapping} of $ X \subseteq \Real^c $ to subsets of $ \Real^n $ is the set $ \mathcal{F}(X) = \{ \vec{y} \in \Real^n \mid \exists \vec{x} \in X, \forall j \in C, x_j = y_j \} $.
\end{definition}

\begin{definition}
	\label{def:compl-hs}
	Let \( H = \halfspace(\vec{a},b) \) be a half-space, the \emph{closed complementary half-space} \( \bar{H} \) is given by \( \bar{H} = \{ \vec{x} \in \Real^d \mid f_{\vec{a}}(\vec{x}) \geq b \} \). 
\end{definition}

\begin{lemma}\label{lem:complement-halfspace}
	If \( \strpElm \in \core(\Game) \) then for all coalitions \( C \subseteq \Ag \) and for all polyhedra $ P^C_j \in \polyset(G^C) $, there is a half-space \( H \in \mathcal{H}^C_j \) such that \( \F((\pay_i(\strpElm))_{i \in C}) \subseteq \F(\bar{H}) \). 
\end{lemma}

\begin{proof}
	Suppose, for the sake of contradiction, that there is a strategy profile \( \strpElm \in \core(\Game) \), coalition \( C \subseteq \Ag \), and polyhedron $ P^C_j $ such that for every half-space \( H \in \mathcal{H}^C_j \) we have \( \F((\pay_i(\strpElm))_{i \in C}) \not\subseteq \F(\bar{H}) \). Thus, it follows that \( \F((\pay_i(\strpElm))_{i \in C}) \subseteq \F(\val(G^C)) \) and there exists a vector \( \vec{x} \in \F(\val(G^C)) \) such that for every player \( i \in C, \) we have \( x_i > \pay_i(\strpElm) \). This implies that there exists a strategy profile \( \strpElm_{C} \) such that for all counter-strategies \( \strpElm_{-C} \) and players \( i \in C \), we have \( \pay_i((\strpElm_{C},\strpElm_{-C})) > \pay_i(\strpElm) \). In other words, there is a beneficial deviation by the coalition \( C \). Therefore, \( \strpElm \) cannot be in the core, leading to a contradiction.
\end{proof}

In essence, Lemma~\ref{lem:complement-halfspace} states that the absence of a beneficial deviation from a strategy profile $ \strpElm $ can be expressed in terms of polyhedral representations and closed complementary half-spaces.
The next lemma, asserts that any value $ \vec{x} \in \Real^c $ enforceable by a coalition $ C $ can also be achieved by the grand coalition $ \Ag $.

\begin{lemma}\label{lem:inclusion}
	For all coalitions \( C \subseteq \Ag \), it holds that \( \val(G^C) \subseteq \proj_C(\val(G^{\Ag}))  \).
\end{lemma}

\begin{proof}
	Suppose, for the sake of contradiction, that there is a vector \( \vec{x} = (x_1,\dots,x_c) \in \val(G^C) \) such that \( \vec{x} \not\in \proj_C(\val(G^{\Ag})) \). This means that there is some strategy profile \( (\strpElm_{C},\strpElm_{-C}) \) and a player \( i \in C \) with \( \pay_i((\strpElm_{C},\strpElm_{-C})) > \pay_i(\strpElm) \) for all \( \strpElm \in \StrSet_{C \cup -C} \). This implies that there is a strategy profile  \( (\strpElm_{C},\strpElm_{-C}) \not\in \StrSet_{C \cup -C} \), i.e., there is a strategy profile that is not included in the set of all strategy profiles, which is a contradiction.
\end{proof}

\begin{figure}[t]
	\begin{subfigure}{0.5\textwidth}
		\centering
		\scalebox{0.7}{
	\begin{tikzpicture}[state/.style={circle, draw, minimum size=1cm}, node distance=1.5cm]
		\node[state, label=below:{$ (0,0,0) $}] (s^0) {\small $s$};
		\node[] (start) [left =1cm of s^0] {};
		\node[state, label=below:{$ (0,2,1) $}] (s^2) [right = of s^0] {\small $m$};
		\node[state, label=below:{$ (2,1,0) $}] (s^1) [above = of s^2] {\small $t$};
		\node[state, label=below:{$ (1,0,2) $}] (s^3) [below = of s^2] {\small $b$};
		
		\draw [-{Latex[width=2mm]}]
		(start) edge node{} (s^0)
		(s^0) edge[] node[right]{\small $(H,H,\ast)$} (s^1)
		(s^0) edge[] node[above]{\small $(*,T,H)$} (s^2)
		(s^0) edge[] node[right]{\small $(T,\ast,T)$} (s^3)
		
		(s^1) edge[loop, out=20, in=340, distance=1cm] node[right]{$*$} (s^1)
		(s^2) edge[loop, out=20, in=340, distance=1cm] node[right]{$*$} (s^2)
		(s^3) edge[loop, out=20, in=340, distance=1cm] node[right]{$*$} (s^3)
		
		(s^0) edge[loop, out=70, in=110, distance=1cm] node[above, align=center]{\small $(H,T,T)$\\ \small $(T,H,H)$} (s^0)
		;
	\end{tikzpicture}
}
	\end{subfigure}
	\begin{subfigure}{0.5\textwidth}
		\centering
		\scalebox{0.7}{
\begin{tikzpicture}
		
	\begin{axis}[xmax=3.1, ymax=3.1, axis x line=middle, axis y line=middle, xmin=-0.5, ymin=-0.5, xtick={0,1,2,3}, xlabel=$x_2$, ylabel=$x_3$, grid=major, grid style=dotted]
		\addplot[] coordinates {(-1, 1) (3, 1)} node[below, xshift=-0.7cm] {$x_3 \leq 1 $};
		\addplot[] coordinates {(2, -1) (2, 3)} node[xshift=-0.6cm, yshift=-0.4cm] {$x_2 \leq 2 $};
		\addplot[fill=gray, draw=none, opacity=0.2] coordinates {(-1, 1) (2, 1) (2,-1) (-1,-1)} \closedcycle;
		\node[label={45:{\small $Q$}},circle,fill,inner sep=2pt] at (axis cs:2,1) {};
		\node[label={0:{\small $R$}},circle,fill,inner sep=2pt] at (axis cs:0,2) {};
		\node[label={120:{\small $P$}},circle,fill,inner sep=2pt] at (axis cs:1,0) {};
		\node[label={90:{\small $S$}},circle,fill=gray,inner sep=2pt] at (axis cs:1,1) {};
		\node[label={90:{\small $H_2 \cap H_3$}}] at (axis cs:0.5,0.5) {};
		
		\draw[shorten >=0.25cm, shorten <=0.25cm,-latex,dashed](axis cs:1,0)--(axis cs:2,1);
		
		\draw[-latex](axis cs:3,1)--(axis cs:3,0.75);
		\draw[-latex](axis cs:2,3)--(axis cs:1.75,3);
	\end{axis}

\end{tikzpicture}

}
	\end{subfigure}
	\caption{Left: Arena for Example~\ref{ex:two}. Right: Graphical representation of $ \val(G^{\{2,3\}}) $. Coordinates $ P,Q,R,S $ corresponds to the set $ \PO(G'^{\Ag}) $. There is a beneficial deviation by $ \{2,3\} $ (dashed arrow) from $ P $ (the $ \{1,2\} $-Pareto optimal value) to $ Q $ (the $ \{2,3\} $-Pareto optimal value), but there is no such a deviation from $ S $.} \label{fig:ex2}
	\vspace{-0.3cm}
\end{figure}
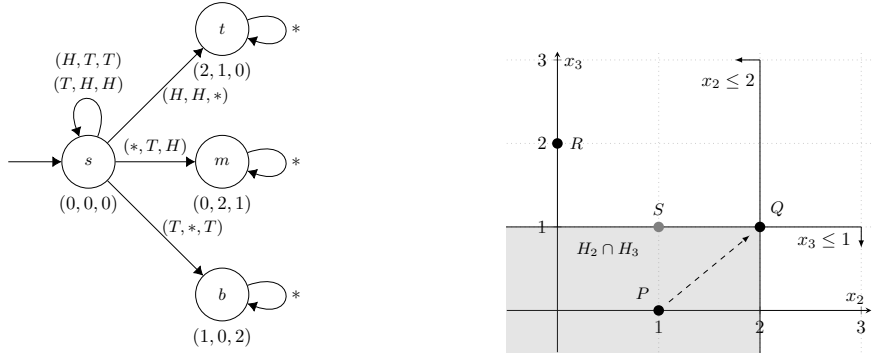

\begin{example}
	\label{ex:two}
	
	Consider a game with $ \Ag = \{1,2,3\} $. The arena is depicted in Figure~\ref{fig:ex2}. 
	Observe that the game has an empty core: if the players stay in $ s $ forever, then $ \{1,2\} $ can beneficially deviate to $ t $. If the play goes to $ t $, then $ \{2,3\} $ can beneficially deviate to $ m $. Similar arguments can be used for $ m $ and $ b $; thus, no strategy profile lies in the core. We can show this using (the contrapositive of) Lemma~\ref{lem:complement-halfspace}: for instance, take a strategy profile $ \strpElm $ that goes to $ t $, and let $ C = \{2,3\} $. Then $ \val(G^{C}) $ can be represented by the intersection of half-spaces $ H_2 = \{ \vec{x} \in \Real^2 \mid x_2 \leq 2 \}$ and $ H_3 = \{ \vec{x} \in \Real^2 \mid x_3 \leq 1 \} $ (see Figure~\ref{fig:ex2} right). Coordinate $ P $ corresponds to $ \strpElm $, and $ \F((\pay_i(\strpElm))_{i \in \{2,3\}}) \not\subseteq \F(\bar{H}_2) $ and $ \F((\pay_i(\strpElm))_{i \in \{2,3\}}) \not\subseteq \F(\bar{H}_3) $. Thus, $ \strpElm $ is not in the core. Now, suppose we modify the game such that $ (\wFun_i(s))_{i \in \Ag} = (1,1,1) $; we obtain  $ \PO(G^{\prime\Ag}) = \{(2,1,0),(0,2,1),(1,0,2),(1,1,1)\} $. Let $ \strpElm' $ be a strategy profile that stays in $ s $ forever (corresponding to $ S $ in Figure~\ref{fig:ex2} right); $ \strpElm' $ is in the core of the modified game, and $ \F((\pay_i(\strpElm'))_{i \in \{2,3\}}) \subseteq \F(\bar{H}_3) $. Indeed for all $ C \subseteq \Ag $ there exists such a half-space. Now if we take the intersection of such half-spaces and the set $ \val(G^{\prime\Ag}) = \downarrow \PO(G^{\prime\Ag}) $, we obtain a non-empty set namely $ \{(1,1,1)\} $ which corresponds to a member of the core $ \strpElm' $. 
\end{example}

From Example~\ref{ex:two}, we observe that a member of the core can be found in the intersection of some set of half-spaces and the set of values enforceable by the grand coalition. We formalise this observation in Theorem~\ref{thm:characterisation}, which provides a {necessary} and {sufficient} condition for the non-emptiness of the core.

\begin{theorem}\label{thm:characterisation}
	The core of a game \(\Game\) is non-empty if and only if there exists a set of half-spaces $ I $ such that 
	\begin{enumerate}
		\item for all coalitions $ C \subseteq \Ag  $ and for all polyhedra $ P^C_j \in \polyset(G^C), I \cap \mathcal{H}^C_j \neq \varnothing $, and 
		\item there exists a polyhedron $ P^{\Ag} \in \polyset(G^N) $ such that $ R = \bigcap_{H \in I} \F(\bar{H}) \cap P^{\Ag} \neq \varnothing $.
	\end{enumerate}
\end{theorem}

\begin{proof}
    From left to right. Suppose that \( \core(\Game) \neq \varnothing \), then there is a strategy profile \( \strpElm \in \core(\Game) \). It follows from Lemma~\ref{lem:complement-halfspace} that for each coalition \( C \subseteq \Ag \) and for each polyhedron $ P^C_j \in \polyset(G^C) $, there exists a half-space \( H \in \mathcal{H}^C_j \) such that \( \F((\pay_i(\strpElm))_{i \in C}) \subseteq \F(\bar{H}) \). Since this holds for each coalition $ C \subseteq \Ag $ and for each polyhedron $ P^C_j \in \polyset(G^C) $, then it is the case that there exists a set of half-spaces $ I $ such that for all coalitions $ C \subseteq \Ag $ and for all polyhedra $ P^C_j \in \polyset(G^C) $ there is a half-space $ H \in I \cap \mathcal{H}^C_j$, and \( (\pay_i(\strpElm))_{i \in \Ag}  \in \bigcap_{H \in I} \F(\bar{H}) \). Furthermore, for each coalition $ C \subseteq \Ag$, it is the case that $ \F((\pay_i(\strpElm))_{i \in C}) \subseteq \F(\val(G^C)) $ and by Lemma~\ref{lem:inclusion}, we have $ (\pay_i(\strpElm))_{i \in C} \in \proj_C(\val(G^{\Ag})) $. Thus, it is also the case that there exists a polyhedron $ P^{\Ag} \in \polyset(G^{\Ag}) $, such that 
    \(  (\pay_i(\strpElm))_{i \in \Ag} \in P^{\Ag} \). Thus, it follows that $ (\pay_i(\strpElm))_{i \in \Ag} \in \bigcap_{H \in I} \F(\bar{H}) \cap P^{\Ag} $ and consequently $ \bigcap_{H \in I} \F(\bar{H}) \cap P^{\Ag} \neq \varnothing $.

    From right to left. Suppose $ R \neq \varnothing $. Take a vector \( \vec{x} \in R \). Since $ \vec{x} \in \bigcap_{H \in I} \F(\bar{H}) $, then for all coalitions \( C \subseteq \Ag \) 
    there is a player $ i \in C $ where $ x_i \geq x_i' $ for some vector \( \vec{x}' \in \F(\PO(G^C)) \). Thus, by the definition of $ C $-Pareto optimality, there exists a player $ i \in C $ that cannot strictly improve its payoff without making other player $ j \in C, j \neq i, $ worse off.
    Thus, for each coalition \( C \subseteq \Ag \) there is no (partial) strategy profile \( \strpElm_{C} \) such that for all counter-strategy profiles \( \strpElm_{-C} \) we have \( \pay_i((\strpElm_{C},\strpElm_{-C})) > x_i \) for every player \( i \in C \). In other words, for each coalition \( C \) and (partial) strategy profile \( \strpElm_{C} \), there is a counter-strategy profile \( \strpElm_{-C} \) that ensures \( \pay_i((\strpElm_{C},\strpElm_{-C})) \leq x_i \). This means that there is no beneficial deviation by the coalition \( C \). Moreover, since \( \vec{x} \in P^{\Ag} \), then we have \( \vec{x} \in \val(G^{\Ag}) \). As such, there exists a strategy profile \( \strpElm \in \StrSet_{\Ag} \) with \( (\pay_i(\strpElm))_{i \in \Ag} \geq \vec{x} \) and \( \strpElm \in \core(\Game) \).
\end{proof}

Using the characterisation of the core from Theorem~\ref{thm:characterisation} above, it follows that if the core is non-empty, then the set \( R \) is a polyhedron $\poly(\lambda)$ for some system of inequations $\lambda$. As such, there exists a vector \( \vec{x} \in R \) whose representation is polynomial in $ n $ and $ \max\{ \size{(\vec{a},b)} \mid (\vec{a},b) \in \lambda \} $~\cite[Theorem 2]{Brenguier2015}. By Lemma~\ref{lem:representation}, it is also the case that $ \max\{ \size{(\vec{a},b)} \mid (\vec{a},b) \in \lambda \} $ is polynomial in the size of the game. Therefore, we obtain the following.

\begin{theorem}\label{thm:polywitness}
	Given a game $ \Game $, if the core is non-empty, then there is $ \strpElm \in \core(\Game) $ such that $ (\pay_i(\strpElm))_{i \in \Ag} $ can be represented polynomially in the size of $ \Game $. 
\end{theorem}

Theorem~\ref{thm:polywitness} plays a crucial role in our approach to solving \nonemptiness and \ecore problems discussed in the next section. It guarantees the existence of a polynomial witness if the core is non-empty, allowing it to be guessed and verified in polynomial time.

\section{Decision problems}\label{secn:decision_problems}
We are now in a position to study each of our decision problems in turn, and establish their complexities. 
We write $ d \in D $ to denote ``$ d $ is a yes-instance of decision problem $ D $''.
Our first problem, called \textsc{Dominated}, serves as an important foundation for studying the other problems. It is formally defined as follows.

\begin{quote}
    \emph{Given}: Game \( \Game \), state \( s \), and vector \( \vec{x} \in \mathbb{Q}^{n} \). \\
    \textsc{Dominated}: Is there a coalition \( C \subseteq \Ag \), and a strategy profile \( \strpElm_{C} \), such that for all counter-strategy profile \( \strpElm_{-C} \), we have \( \pay_i(\pi((\strpElm_{C},\strpElm_{-C}),s)) > x_i \) for each \( i \in C \)?
\end{quote}

\begin{restatable}{theorem}{THMdominated}
	\label{thm:dominated}
	\textsc{Dominated} is \( \SigmaPTwo \)-complete.
\end{restatable}

\algdominated

\begin{proof}
Observe that an instance $ (\Game,s,\vec{x}) \in $ \textsc{Dominated} has a witness vector $ (x_i')_{i \in C} $ that lies in the intersection of a polyhedron $ P^C \in \polyset(G^C,s) $ and the set $ \{ \vec{y} \in \Real^c \mid \forall i \in C: y_i \geq x_i \} $. Such an intersection forms a polyhedron $ \poly(\lambda) $, definable by a system of linear inequalities $ \lambda $. By Lemma~\ref{lem:representation}, each $ (\vec{a},b) \in \lambda $ has polynomial representation in the size of $ G^C $. Therefore, $ (x_i')_{i \in C} $ has a representation that is polynomial in the size of $ \Game $. To solve \textsc{Dominated}, we provide Algorithm~\ref{alg:dominated}. The correctness follows directly from the definition of \textsc{Dominated}. For the upper bound: since $ (x_i')_{i \in C} $ is of polynomial size, line 1 can be done in \( \np \). In line 2, we have subprocedure \textsc{Sequentialise} that builds and returns sequentialisation of $ \Game $ w.r.t. coalition $ C $; this can be done in polynomial time. Finally, line 3 is in \( \conp \)~\cite[Theorem 3, Lemma 6]{Velner2015}. Therefore, the algorithm runs in \( \np^{\conp} = \SigmaPTwo \).

    For the lower bound, we reduce from \( \QSAT_2 (3\DNF) \) (satisfiability of quantified Boolean formulae with 2 alternations and 3DNF clauses). The complete proof is included in the technical appendix.

	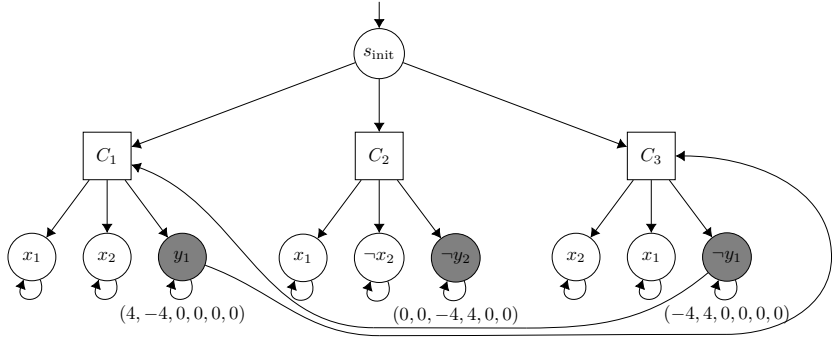
\begin{figure*}[t]
		\centering
		\scalebox{0.7}{
\begin{tikzpicture}[state/.style={circle, draw, minimum size=0.9cm}, dbl/.style={circle, draw, fill=gray, minimum size=0.9cm}, sqr/.style={rectangle, draw, minimum size=0.9cm},
	dsqr/.style={fill=gray, rectangle, draw, minimum size=0.9cm},
	bezier bounding box=true]
		\node[state] (s_init) {$s_{\text{init}}$};
		\coordinate [above =0.5cm of s_init] (start);
                
		\node[sqr] (C_2) [below = of s_init] {$C_2$};
		\node[sqr] (C_1) [left = 4.2cm of C_2] {$C_1$};
		\node[sqr] (C_3) [right = 4.2cm of C_2] {$C_3$};

                \node[state] (C_1-x_2) [below = of C_1] {$x_2$};
                \node[state] (C_1-x_1) [left = 0.5cm of C_1-x_2] {$x_1$};
                \node[dbl] (C_1-x_3) [right = 0.5cm of C_1-x_2] {$y_1$};

                \node[state] (C_2-x_2) [below = of C_2] {$\lnot x_2$};
                \node[state] (C_2-x_1) [left = 0.5cm of C_2-x_2] {$x_1$};
                \node[dbl] (C_2-x_3) [right = 0.5cm of C_2-x_2] {$\lnot y_2$};

                \node[state] (C_3-x_2) [below = of C_3] {$x_1$};
                \node[state] (C_3-x_1) [left = 0.5cm of C_3-x_2] {$x_2$};
                \node[dbl] (C_3-x_3) [right = 0.5cm of C_3-x_2] {$\lnot y_1$};

                \coordinate [below = of C_2-x_2] (C_1-x_3-curve-stop-1);
                \coordinate [below = of C_3-x_3] (C_1-x_3-curve-stop-2) ;
                \coordinate [right = 1.5cm of C_3-x_3] (C_1-x_3-curve-stop-3);

                \coordinate [below right = of C_2-x_3] (C_3-x_3-curve-stop-1);
                \coordinate [below right =of C_2-x_1] (C_3-x_3-curve-stop-2);

		\draw[-{Latex[width=2mm]}] (start) --  (s_init);
		\draw[-{Latex[width=2mm]}] (s_init) --  (C_1);
		\draw[-{Latex[width=2mm]}] (s_init) --  (C_2);
		\draw[-{Latex[width=2mm]}] (s_init) --  (C_3);

                \draw[-{Latex[width=2mm]}] (C_1) -- (C_1-x_1);
                \draw[-{Latex[width=2mm]}] (C_1) -- (C_1-x_2);
                \draw[-{Latex[width=2mm]}] (C_1) -- (C_1-x_3);

                \draw[-{Latex[width=2mm]}] (C_2) -- (C_2-x_1);
                \draw[-{Latex[width=2mm]}] (C_2) -- (C_2-x_2);
                \draw[-{Latex[width=2mm]}] (C_2) -- (C_2-x_3);

                \draw[-{Latex[width=2mm]}] (C_3) --  (C_3-x_1);
                \draw[-{Latex[width=2mm]}] (C_3) -- (C_3-x_2);
                \draw[-{Latex[width=2mm]}] (C_3) --  (C_3-x_3);

		\draw[-{Latex[width=2mm]}] (C_1-x_1) edge[loop, out=290, in=250, distance=0.5cm] (C_1-x_1);
		\draw[-{Latex[width=2mm]}] (C_1-x_2) edge[loop, out=290, in=250, distance=0.5cm] (C_1-x_2);
		\draw[-{Latex[width=2mm]}] (C_1-x_3) edge[loop, out=290, in=250, distance=0.5cm] node[below]{$(4, -4,0,0,0,0)$} (C_1-x_3);

                \draw[-{Latex[width=2mm]}] (C_2-x_1) edge[loop, out=290, in=250, distance=0.5cm] (C_2-x_1);
		\draw[-{Latex[width=2mm]}] (C_2-x_2) edge[loop, out=290, in=250, distance=0.5cm] (C_2-x_2);
		\draw[-{Latex[width=2mm]}] (C_2-x_3) edge[loop, out=290, in=250, distance=0.5cm] node[below]{$(0,0,-4,4,0,0)$} (C_2-x_3);

                \draw[-{Latex[width=2mm]}] (C_3-x_1) edge[loop, out=290, in=250, distance=0.5cm] (C_3-x_1);
		\draw[-{Latex[width=2mm]}] (C_3-x_2) edge[loop, out=290, in=250, distance=0.5cm] (C_3-x_2);
		\draw[-{Latex[width=2mm]}] (C_3-x_3) edge[loop, out=290, in=250, distance=0.5cm] node[below]{$(-4,4,0,0,0,0)$} (C_3-x_3);

                \draw[-{Latex[width=2mm]}] (C_1-x_3-curve-stop-3) to [out=90, in=0] (C_3);
                \draw[-{Latex[width=2mm]}] (C_3-x_3-curve-stop-2) to [out=180, in=340] (C_1);

                \draw[] (C_1-x_3) to [out=340,in=180] (C_1-x_3-curve-stop-1);
                \draw[] (C_1-x_3-curve-stop-1) to [out=0,in=180] (C_1-x_3-curve-stop-2);
                \draw[] (C_1-x_3-curve-stop-2) to [out=0,in=270] (C_1-x_3-curve-stop-3);

                \draw[] (C_3-x_3) to [out=220,in=0] (C_3-x_3-curve-stop-1);
                \draw[] (C_3-x_3-curve-stop-1) to [out=180,in=0] (C_3-x_3-curve-stop-2);
		;
\end{tikzpicture}
}
		\caption{
			The game arena of $ \Game^\Phi $. White circle states are controlled by $ E $, square by $ A $, grey circles $ y_1, \neg y_1, \neg y_2 $ by players $ 1, 2, 4 $, respectively. The weight function is given as vectors shown below the states, and states without vectors have $ \vec{0} $. 
			Furthermore, each grey circle state also has a transition to $ \sink $, which is not shown in the figure.
		}
		\label{fig:qsat}
		\vspace{-0.3cm}
	\end{figure*}
	
    To illustrate the reduction, consider the formula $$ \Phi = \exists x_1 \exists x_2 \forall y_1 \forall y_2 (x_1 \wedge x_2 \wedge y_1) \vee (x_1 \wedge \neg x_2 \wedge \neg y_2) \vee (x_1 \wedge x_2 \wedge \neg y_1). $$
    We build a corresponding game $ \Game^\Phi $ such that $ (\Game^\Phi, \sinit, (-1,-1,-1,-1,-1,0)) = \chi \in $ \textsc{Dominated} if and only if $ \Phi $ is satisfiable.
    To this end, we construct the game \( \Game^\Phi \) in Figure~\ref{fig:qsat} with $ \Ag = \{1,2,3,4,E,A\} $ and the weight function given as vectors, such that for a given vector $ (w_1,...,w_6) $ in state $ s $, $ \wFun_{i}(s) = w_i, i \in \{1,2,3,4\} $ and $ \wFun_{E}(s) = w_5, \wFun_{A}(s) = w_6$. The $ \sink $ (not shown) only has transition to itself and its weights is given by the vector $ (-1,-1,-1,-1,-1,0) $. The intuition is that if $ \Phi $ is satisfiable, then there is a joint strategy $ \strpElm_{C} $ by $ C = \Ag \setminus \{A\} $ that guarantees a payoff of $ 0 $ for each $ i \in C $. If $ \Phi $ is not satisfiable, then $ A $ has a strategy that visits some state $ y_k $ (resp. $ \neg y_k $) infinitely often and player $ 2k-1 $ (resp. $ 2k $) gets payoff $ < -1 $. Since $ y_k $ (resp. $ \neg y_k $) is controlled by $ 2k-1 $ (resp. $ 2k $), then the player will deviate to $ \sink $, and $ \chi \notin \dominated $. On the other hand, if $ \chi \in \dominated $, then there exists a strategy $ \strpElm_{C} $ which guarantees that the play: (a) ends up in some state $ x_k $ or $ \neg x_k $, or (b) visits both $ y_k $ and $ \neg y_k $ infinitely often. For the former, it means that there is a clause with only $ x $-literals, and the latter implies that for all (valid) assignments of $ y $-literals, there is an assignment for $ x $-literal that makes at least one clause evaluate to true. Both cases show that $ \Phi $ is satisfiable.
    Now, notice that the formula \( \Phi \) is satisfiable: take the assignment that set \( x_1 \) and \( x_2 \) to be both true. Indeed, $ \chi \in \textsc{Dominated} $: the coalition $\{1,2,3,4,E\}$ have a strategy that results in payoff vector \( \vec{0} \), e.g., take a strategy profile that corresponds to the cycle \( ({C_1}{y_1}{C_3}{\neg y_1})^{\omega} \).
\end{proof}

Our next problem \ebendev simply asks if a given game has a beneficial deviation from a provided strategy profile:

\begin{quote}
    \emph{Given}: Game \( \Game \), strategy profile \(\strpElm\). \\
    {\ebendev}: Does there exist some coalition \(C \subseteq \Ag\) such that \(C\) has a beneficial deviation from \(\strpElm\)?
\end{quote}

Notice that \ebendev is closely related to \dominated. Firstly, we fix $ s $ to be the initial state. Secondly, instead of a vector, we are given a strategy profile. If we can compute the payoff induced by the strategy profile, then we can immediately reduce \ebendev to \dominated. \cite{steeples2021mean} studies this problem in the memoryless setting, but the approach presented there (i.e., by ``running'' the strategy profile and calculating the payoff vectors) does not generalise to finite-state strategies $ \strpElm $ as the lasso $ \pi(\strpElm) $ may be of exponential size. To illustrate this, consider a profile $ \strpElm $ that acts like a binary counter. We have $ \card{\strpElm} $ that is of polynomial size, but when we run the profile, we obtain an exponential number of step before we encounter the same configuration of game and strategies states. 
However, in order to compute the payoff vector of a finite-state strategy profile $ \strpElm $, we only need polynomial space. First, we recall that for deterministic, finite-state strategies, the path $ \pi(\strpElm) $ is ultimately periodic (i.e., a lasso-path). As such, there exist $ (s^k,\jact^k) $ and $ (s^l,\jact^l) $ with $ l > k $ and $ \config(\strpElm,k) = \config(\strpElm,l) $.
With this observation, computing the payoff vector can be done by Algorithm~\ref{alg:compute}.

\algcompute

Line 1 can be done non-deterministically in polynomial space. In line 2, we have \textsc{ComputeIndex} subprocedure that computes and returns $ k,l $. This procedure is also in polynomial space: we run the profile $ \strpElm $ from $ \sinit $ and in each step only store the current configuration; for the first time we have $ \config(\strpElm,t) = (s^j,q^j_1, \ldots, q^j_n) $, assign $ k = t $, and the second time $ \config(\strpElm,t') = (s^j,q^j_1, \ldots, q^j_n) $, assign $ l = t' $, and we are done. Note that this subprocedure returns the smallest pair of $ k,l $. Line 3 is in polynomial time. So, overall we have a function problem that can be solved in \npspace, and by Savitch's theorem we obtain the following.

\begin{lemma}\label{lem:payoff}
    For a given \( \Game \) and \( \strpElm \), the payoff vector \( (\pay_i(\strpElm))_{i \in \Ag} \) can be computed in \pspace. 
\end{lemma}

This puts us in position to determine the complexity of \ebendev as follows.

\begin{restatable}{theorem}{THMbendev}
	\label{thm:e-ben-dev}
	\ebendev is $ \pspace $-complete.
\end{restatable}

\begin{proof}
    To solve \ebendev, we reduce it to \dominated as follows. First, using Algorithm~\ref{alg:compute} we compute \( (\pay_i(\strpElm))_{i \in \Ag} \) in \pspace (Lemma~\ref{lem:payoff}). Then, using Algorithm~\ref{alg:dominated} we can check whether \( (\Game,\sinit,(\pay_i(\strpElm))_{i \in \Ag}) \in \dominated \). Since $ \SigmaPTwo \subseteq \pspace $, \ebendev can be solved in \pspace.
	For the lower bound, we reduce from the non-emptiness problem of intersection of automata that is known to be \pspace-hard~\cite{4567949}. The full proof is provided in the technical appendix.
\end{proof}

Another decision problem that is naturally related to the core is asking whether a given strategy profile $ \strpElm $ is in the core of a given game. The problem is formally stated as follows.

\begin{quote}
	\emph{Given}: Game \( \Game \) and strategy profile \( \strpElm \). \\
	\textsc{Membership}: Is it the case that \( \strpElm \in \core(\Game) \)?
\end{quote}

Observe that we can immediately see the connection between \ebendev and \textsc{Membership}: they are essentially dual to each other. Therefore, we immediately obtain the following lemma.

\begin{lemma}\label{lem:dual}
    For a given game \( \Game \) and strategy profile \( \strpElm \), it holds that \( \strpElm \in \core(\Game) \) if and only if \( (\Game,\strpElm) \in \) \ebendev.
\end{lemma}

Using Lemma~\ref{lem:dual} and the fact that co-\pspace = \pspace, we obtain the following theorem.

\begin{theorem}\label{thm:membership}
	\textsc{Membership} is \pspace-complete.
\end{theorem}

In rational verification, we check which temporal logic properties are satisfied by a game's stable outcomes. Two key decision problems are formally defined as follows.

\begin{quote}
	\emph{Given}: Game \(\Game\), formula \(\phi\). \\
	\textsc{E-Core}: Is it the case that there exists some \(\vec{\sigma}\in\core(\Game)\) such that \(\vec{\sigma}\models \phi\)?\\
	\textsc{A-Core}: Is it the case that for all \(\vec{\sigma}\in\core(\Game)\), we have \(\vec{\sigma}\models \phi\)?
\end{quote}

To illustrate the decision problems above, let us revisit Example~\ref{example}. Consider a query of \acore for Example~\ref{example} with property $ \phi = \always \sometime l \wedge \always \sometime r $. Such a query will return a positive answer, i.e., \textit{every} strategy profile that lies in the core satisfies $ \phi $.

Another key decision problem in rational verification is determining whether a given game has any stable outcomes. This involves checking if the game has a non-empty core. 

\begin{quote}
    \emph{Given}: Game \( \Game \). \\
    \textsc{Non-Emptiness}: Is it the case that \( \core(\Game) \neq \varnothing \)?
\end{quote}

As demonstrated in Example~\ref{ex:two}, there exist mean-payoff games with an empty core---this is in stark contrast to the dichotomous preferences setting (cf.~\cite{Gutierrez2019,HLNW21}). As such, \nonemptiness problem is non-trivial in mean-payoff games. 

\begin{table}[t]
	\begin{tabular}{l}
		\hspace{-0.3cm}
		\begin{minipage}{0.5\textwidth}
			\algnonemptiness
		\end{minipage}%
	\hspace{0.5cm}
	 \begin{minipage}{0.5\textwidth}
	 	\centering
		\begin{figure}[H]
			\scalebox{0.8}{
	\hspace*{1cm}
	\begin{tikzpicture}
		\draw (0,0) -- (5,0) -- (5,5) -- (0,5) -- (0,0) node[xshift=0.2cm,yshift=0.3cm] {$ \Game $};
		\coordinate (c) at (2.5,2.5);
		\draw[rounded corners=1mm, line width=1mm, gray!40, fill] (c) \irregularcircle{2.1cm}{2mm} node [xshift=-2cm,yshift=1.5cm] {\Large $  $};
		\node at (1,1.6)[circle,fill=black,label=above:{\Large $ \sinit $}]{};
		\node at (2.5,3.5) [align=center] { $ \forall {s \in S}$ \\ $ (\Game,s,\vec{x}) \notin \dominated$} ;
		\node at (2.5,1) {\large $ \Game{[S]} $} ;
		\path
			(1,1.6) edge[dashed,line width=0.5mm] (3,1.6)
			(3,1.6) edge [dashed,-latex,shorten >=1pt,loop,in=70,out=15,distance=2cm,line width=0.5mm] node[xshift=-1.2cm]{$ \pi $} (2,1.6) 
		;
		
	\end{tikzpicture}
	
}
			\caption{Illustration for solving \ecore.}
			\label{fig:ecore}
		\end{figure}
	 \end{minipage}
	\end{tabular}
\vspace*{-0.5cm}
\end{table}

To solve \nonemptiness, it is important to recall the following two results. Firstly, if a game $ \Game $ has a non-empty core, then there is a payoff vector $ \vec{x} $ resulting from $ \strpElm \in \core(\Game) $ whose representation is polynomial (Theorem~\ref{thm:polywitness}). Secondly, if $ \vec{x} $ is a witness for the core, then \( (\Game,\sinit,\vec{x}) \notin \) \dominated.
With these observations, solving \textsc{Non-Emptiness} can be done by Algorithm~\ref{alg:nonemptiness}.
The subprocedure in line 1 is polynomial. Line 2 is in \np (Theorem~\ref{thm:polywitness}) and we call $ \SigmaPTwo $ oracle for line 3. Thus, Algorithm~\ref{alg:nonemptiness} runs in $ \SigmaPThree $. For hardness, we reduce from \( \QSAT_3 (3\CNF) \) (satisfiability of quantified Boolean formulae with 3 alternations and 3CNF clauses). The reduction has a similar flavour to the one used in Theorem~\ref{thm:dominated}, albeit a bit more involved. The complete proof is included in the technical appendix.

\begin{restatable}{theorem}{THMemptiness}
	\label{thm:non-emptiness}
	\textsc{Non-Emptiness} is \( \SigmaPThree \)-complete.
\end{restatable}

		\algecore

Now we turn our attention to \textsc{E-Core}. Observe that for a game \( \Game \) and a LTL specification \( \phi \), a witness to \textsc{E-Core} would be a path \( \pi \) such that \( (\pay_i(\pi))_{i \in \Ag} \geq (\pay_i(\strpElm))_{} \) for some \( \strpElm \in \core(\Game) \), and \( \pi \models \phi \). Furthermore, a (satisfiable) LTL formula \( \phi \) has an ultimately periodic model of size \( 2^{O(\card{\phi})} \)~\cite{Sistla1985}. Thus, the size of representation of \( \pay_i(\pi) \) is at most \( \log_2(\card{W} \cdot 2^{O\card{\phi}}) \), where \( W \) is the maximal absolute value appearing in the weights in \( \Game \), i.e., \( W = \max\{ \card{\wFun_i(s)} \mid i \in \Ag, s \in \St \} \). To solve \textsc{E-Core} with a LTL specification \( \phi \) we use Algorithm~\ref{alg:ecore}. An intuitive illustration is provided in Figure~\ref{fig:ecore}. We begin by guessing a vector $ \vec{x} \in \Rat^n $ and a set of states $ S \subseteq \St $, such that for every $ s \in S, (\Game,s,\vec{x}) \notin \dominated $. Next, we obtain a (sub-)game $ \Game{[S]} $ (shaded area) by removing all states $ s \notin S $ and edges leading to those removed states. In this new game $ \Game{[S]} $, we identify the lasso path $ \pi $ with $ \pay_i(\pi) \geq x_i $ for all $ i \in \Ag $ and $ \pi \models \phi $. This path corresponds to a strategy profile in the core since there is no beneficial deviation by any $ C \subseteq \Ag $ in any state in it.

Line 1 is in polynomial time. Line 2 can be done in \np (Theorem~\ref{thm:polywitness}). In line 3, we can use Algorithm~\ref{alg:dominated} with a slight modification: the state $ s $ is \textit{not} given as part of the input, but \textit{included in the first guess} in the algorithm. Clearly, the modified algorithm still runs in \( \SigmaPTwo \). In line 4, we have the subprocedure \textsc{UpdateArena} that returns $ \Game{[S]} $; this can be done in polynomial time. For line 5, consider the LTL\( ^{lim\Sigma} \) formula $ \psi \vcentcolon= \phi \wedge \bigwedge_{i \in \Ag} (\MP(\wFun_{i}) \geq x_i). $ Observe that a path in \( \Game{[S]} \) satisfying the formula \( \psi \) corresponds to a strategy profile \( \strpElm \) such that in every state \( s \) in \( \pi(\strpElm) \), \( (\Game, s, (\pay_i(\strpElm))_{i \in \Ag}) \notin \)  \textsc{Dominated}. Thus, it follows that \( \strpElm \in \core(\Game) \), and additionally, \( \pi(\strpElm) \models \phi \). Finding such a path corresponds to (existential) model checking \( \psi \) against the underlying arena of \( \Game{[S]} \)
---this can be done in \pspace~\cite{Boker2014}. Hardness directly follows from setting \( \wFun_i(s) = 0 \) for all \( i \in \Ag \) and \( s \in \St \). For \textsc{A-Core}, observe that the problem is exactly the dual of \textsc{E-Core}, and since co-\pspace = \pspace, we have the following theorem.

\begin{theorem}\label{thm:ltl-core}
    The \textsc{E-Core} and \textsc{A-Core} problems with LTL specifications are \pspace-complete.
\end{theorem}

\subparagraph*{\textsc{E/A-Core} with \GRone Specifications.}
The main bottleneck in Algorithm~\ref{alg:ecore} for LTL specifications is in line 5, where the model checking of LTL\( ^{lim\Sigma} \) formula occurs. This can be avoided by considering classes of properties with easier model checking problem. In this section, we address \textsc{E/A-Core} with \GRone specifications\footnote{We could use any other ``easy'' fragment of LTL to avoid this bottleneck, as we will discuss later.}. The approach is similar to that in \cite[Theorem 18]{ummels2011the}. The main idea is to define a linear program \( \mathcal{L} \) such that it has a feasible solution if and only if the condition in line 5 of Algorithm~\ref{alg:ecore} is met.

To this end, first recall that a \GRone formula \(\varphi\) has the following form 
   \[ \varphi = \bigwedge_{l = 1}^{m} \always \sometime \psi_{l} \to \bigwedge_{r = 1}^{n} \always \sometime \theta_{r}\text{,} \]
and let \(V(\psi_{l})\) and \(V(\theta_r)\) be the subset of states in \(\Game\) that satisfy the Boolean combinations \(\psi_{l}\) and \(\theta_{r}\), respectively.  Observe that property \(\varphi\) is satisfied over a path \(\pi\) if, and only if, either \(\pi\) visits every \(V(\theta_r)\) infinitely many times or visits some of the \(V(\psi_{l})\) only a finite number of times. To check the satisfaction of $ \bigwedge_{l = 1}^{m} \always \sometime \psi_{l} $ we define linear programs \(\mathcal{L}(\psi_{l})\) such that it admits a solution if and only if there is a path $ \pi $ in $ \Game{[S]} $ such that \(\pay_i(\pi) \geq x_i \) for every player \(i\) and visits \(V(\psi_{l})\) only \emph{finitely many times}. Similarly, for $ \bigwedge_{r = 1}^{n} \always \sometime \theta_{r} $, define a linear program \(\mathcal{L}(\theta_{1}, \ldots, \theta_{n})\) that admits a solution if and only if there exists a path \(\pi\) in \( \Game{[S]} \) such that \(\pay_i(\pi) \geq x_i\) for every player \(i\) and visits every \(V(\theta_{r})\) \emph{infinitely many times}. Both linear programs are polynomial in the size of \(\Game\) and \(\phi\), and at least one of them admits a solution if and only if $ \phi $ is satisfied in some path in $ \Game{[S]} $. Therefore, given \( \Game{[S]} \) and \GRone formula \( \phi \) it is possible to check in polynomial time whether \( \phi \) is satisfied by a suitable path $ \pi $ in \( \Game{[S]} \). The detailed construction is provided in the technical appendix.

Therefore, to solve \textsc{E-Core} with \GRone specifications, we can use Algorithm~\ref{alg:ecore} with polynomial time check for line 5. Thus, it follows that \textsc{E-Core} with \GRone specifications can be solved in $ \SigmaPThree $. The lower bound follows directly from hardness result of $ \textsc{Non-Emptiness} $ by setting $ \phi = \top $. Moreover, since $ \acore $ is the dual of $ \ecore $, we obtain the following theorem.

\begin{restatable}{theorem}{THMgrecore}
	\label{thm:gr-ecore}
    The \textsc{E-Core} and \acore problems with \GRone specifications are \( \SigmaPThree \)-complete and $ \PiPThree $-complete, respectively.
\end{restatable}

\subparagraph*{E/A-Core with Other Specifications.}
We conclude this section by noting that the approach presented here for solving \textsc{E/A-Core} problem is easily generalisable to different types of specification languages without incurring additional computational costs. For instance, the approach for $ \GRone $ is directly applicable to the $ \omega $-regular specifications considered in~\cite{steeples2021mean}. Furthermore, Algorithm~\ref{alg:ecore} can also be easily adapted for LTL fragments whose witnesses are of polynomial size w.r.t. $ \Game $ and $ \varphi $~\cite{demri2002complexity,markey2004past}. This can be done by (1) guessing a witness $ \pi $ in line 2 and (2) checking whether $ \pi \models \phi $ and $ \pay_i(\pi) \geq x_i $ for all $ i \in \Ag $ in line 5, resulting in the same complexity classes as stated in Theorem~\ref{thm:gr-ecore}.

\begin{figure}[t]
	\centering
	\begin{tabular}{l r r | r}
		\toprule
		Problem             & Finite Memory                       & Memoryless & NE\phantom{123}\\
		\midrule
		\textsc{Dominated}           & \(\Sigma^{\text{P}}_2\)-c (Thm.~\ref{thm:dominated}) &  & \\
		\(\exists\)-\textsc{Ben-Dev} & \pspace-c (Thm.~\ref{thm:e-ben-dev}) & \np-c&\\
		\textsc{Membership}          & \pspace-c (Thm.~\ref{thm:membership}) & \conp-c& \\
		\textsc{Non-Emptiness}       & \(\Sigma^{\text{P}}_3\)-c (Thm.~\ref{thm:non-emptiness}) & $ \SigmaPTwo $& \np-c\\
		\textsc{E-Core} with LTL spec.      & \pspace-c (Thm.~\ref{thm:ltl-core}) &  & \pspace-c\\
		\textsc{A-Core} with LTL spec.      & \pspace-c (Thm.~\ref{thm:ltl-core}) &  & \pspace-c\\
		\textsc{E-Core} with GR(1) spec.    & \(\Sigma^{\text{P}}_3\)-c (Thm.~\ref{thm:gr-ecore}) & $\SigmaPTwo$& \np-c\\
		\textsc{A-Core} with GR(1) spec.    & \(\Pi^{\text{P}}_3\)-c (Thm.~\ref{thm:gr-ecore}) & $\PiPTwo$& \conp-c\\
		\bottomrule
	\end{tabular}
	\caption{Summary of complexity results. 
		The NE column shows complexity results for the corresponding decision problems with NE. Complexity results for decision problems related to the core in the memoryless setting can be found in~\protect\cite{steeples2021mean}, whereas for NE in~\protect\cite{ummels2011the,gutierrez2019on}. 
	}
	\label{fig:summary_of_results}
	\vspace{-0.3cm}
\end{figure}

\section{Concluding remarks}\label{secn:concluding_remarks}
In this paper, we present a novel characterisation of the core of cooperative concurrent mean-payoff games using discrete geometry techniques
which differs from previous methods that relied on logical characterisation and punishment/security values~\cite{steeples2021mean,gutierrez2023cooperative}. 
We have also determined the exact complexity of several related decision problems in rational verification. Our results and other related results from previous work are summarised in Figure~\ref{fig:summary_of_results}. 

It is interesting to note that \nonemptiness of the core is two rungs higher up the polynomial hierarchy from its NE counterpart. This seems to be induced by the fact that for a given deviation, the punishment/counter-strategy is not static as in NE. It is also worth mentioning that generalising to finite-memory strategies (second column) results in an increase in complexity classes compared to the memoryless setting (third column). In particular, \ebendev and \textsc{Membership} jump significantly from \np-complete and \conp-complete, respectively, to \pspace-complete. Furthermore, and rather surprisingly, in the finite memory setting, \ebendev and \textsc{Membership} are harder than \nonemptiness, which sharply contrasts with the memoryless setting. This seems to be caused by the following: Algorithm~\ref{alg:nonemptiness} for \nonemptiness is ``non-constructive'', in the sense that we only care about the existence of a strategy profile that lies in the core without having to explicitly construct one. On the other hand, with \textsc{Membership}, we have to calculate the payoff from a compact representation of a given strategy profile, which requires us to ``unpack'' the profile.

Our characterisation of the non-emptiness of the core (Theorem~\ref{thm:characterisation}) provides a way to ensure that the core always has a polynomially representable witness. However, it would be interesting to establish the sufficient and necessary conditions in a broader sense. Previous work has addressed the sufficient and necessary conditions for the non-emptiness of the core in non-transferable utility (NTU) games. For example, \cite{scarf1971existence} showed that the core of an NTU game is non-empty when the players have continuous and quasi-concave utility functions. \cite{uyanik2015nonemptiness} relaxed the continuity assumption (which aligns more closely with the setting in this paper) and achieved a result similar to \cite{scarf1971existence}. However, their game models differ from ours, and the results do not automatically apply. We conjecture that a similar condition, namely the quasi-concavity of utility funtions, plays a vital role in the non-emptiness of the core in concurrent multi-player mean-payoff games. Nevertheless, this still needs to be formally proven and would make for interesting future work.

As previously mentioned, a key difference between the core of concurrent multi-player mean-payoff games and games with dichotomous preferences is that the former may have an empty core. This raises the question: what can we do when the core is empty? One might want to introduce stability, thereby making the core non-empty. One approach, which relates to the above conjecture, involves modifying the utility functions, for instance, through subsidies or rewards \cite{GNPW19,Almagor2015}. Another approach is to introduce norms \cite{perelli2019enforcing}. This is an area for future exploration.

It would also be interesting to generalise the current work to decidable classes of imperfect information mean-payoff games~\cite{Degorre2010}. Another potential avenue is to relax the concurrency, for instance, by making agents loosely coupled. A different but intriguing direction would be to investigate the possibility of using our construction and characterisation here to extend ATL* with mean-payoff semantics.

\bibliography{references}

\begin{thebibliography}{10}

\bibitem{abate2021rational}
Alessandro Abate, Julian Gutierrez, Lewis Hammond, Paul Harrenstein, Marta
  Kwiatkowska, Muhammad Najib, Giuseppe Perelli, Thomas Steeples, and Michael
  Wooldridge.
\newblock Rational verification: game-theoretic verification of multi-agent
  systems.
\newblock {\em Applied Intelligence}, 51(9):6569--6584, 2021.

\bibitem{Almagor2015}
S.~Almagor, G.~Avni, and O.~Kupferman.
\newblock {Repairing Multi-Player Games}.
\newblock In {\em CONCUR}, volume~42 of {\em LIPIcs}, pages 325--339. Schloss
  Dagstuhl, 2015.

\bibitem{ADMW09}
Rajeev Alur, Aldric Degorre, Oded Maler, and Gera Weiss.
\newblock On omega-languages defined by mean-payoff conditions.
\newblock In Luca de~Alfaro, editor, {\em Foundations of Software Science and
  Computational Structures}, pages 333--347, Berlin, Heidelberg, 2009. Springer
  Berlin Heidelberg.

\bibitem{Alur2002}
Rajeev Alur, Thomas~A. Henzinger, and Orna Kupferman.
\newblock Alternating-time temporal logic.
\newblock {\em J. ACM}, 49(5):672--713, September 2002.
\newblock \href {https://doi.org/10.1145/585265.585270}
  {\path{doi:10.1145/585265.585270}}.

\bibitem{aumann1961core}
Robert~J Aumann.
\newblock The core of a cooperative game without side payments.
\newblock {\em Transactions of the American Mathematical Society},
  98(3):539--552, 1961.

\bibitem{berthon2020mixing}
Rapha{\"e}l Berthon, Shibashis Guha, and Jean-Fran{\c{c}}ois Raskin.
\newblock Mixing probabilistic and non-probabilistic objectives in markov
  decision processes.
\newblock In {\em Proceedings of the 35th Annual ACM/IEEE Symposium on Logic in
  Computer Science}, pages 195--208, 2020.

\bibitem{elisabertino2020artificial}
Elisa Bertino, Finale Doshi-Velez, Maria Gini, Daniel Lopresti, and David
  Parkes.
\newblock Artificial intelligence and cooperation.
\newblock Technical Report White Paper 4, Computing Community Consortium,
  Washington, D.C., 10 2020.
\newblock URL:
  \url{https://cra.org/ccc/wp-content/uploads/sites/2/2020/11/AI-and-Cooperation.pdf}.

\bibitem{Bloem2012}
Roderick Bloem, Barbara Jobstmann, Nir Piterman, Amir Pnueli, and Yaniv Sa'ar.
\newblock Synthesis of reactive(1) designs.
\newblock {\em J. Comput. Syst. Sci.}, 78(3):911--938, May 2012.
\newblock \href {https://doi.org/10.1016/j.jcss.2011.08.007}
  {\path{doi:10.1016/j.jcss.2011.08.007}}.

\bibitem{Boker2014}
Udi Boker, Krishnendu Chatterjee, Thomas~A. Henzinger, and Orna Kupferman.
\newblock Temporal specifications with accumulative values.
\newblock {\em ACM Trans. Comput. Log.}, 15(4):1--25, August 2014.
\newblock \href {https://doi.org/10.1145/2629686} {\path{doi:10.1145/2629686}}.

\bibitem{bondareva1963some}
Olga~N Bondareva.
\newblock Some applications of linear programming methods to the theory of
  cooperative games.
\newblock {\em Problemy kibernetiki}, 10(119):139, 1963.

\bibitem{bouyer2015pure}
Patricia Bouyer, Romain Brenguier, Nicolas Markey, and Michael Ummels.
\newblock Pure {Nash} equilibria in concurrent deterministic games.
\newblock {\em Log. Meth. Comput. Sci.}, 11(2), June 2015.
\newblock \href {https://doi.org/10.2168/lmcs-11(2:9)2015}
  {\path{doi:10.2168/lmcs-11(2:9)2015}}.

\bibitem{bouyer2023reasoning}
Patricia Bouyer, Orna Kupferman, Nicolas Markey, Bastien Maubert, Aniello
  Murano, and Giuseppe Perelli.
\newblock Reasoning about quality and fuzziness of strategic behaviors.
\newblock {\em ACM Transactions on Computational Logic}, 24(3):1--38, 2023.

\bibitem{Brenguier2015}
Romain Brenguier and Jean-François Raskin.
\newblock {Pareto} curves of multidimensional mean-payoff games.
\newblock In Daniel Kroening and Corina~S. Păsăreanu, editors, {\em Computer
  Aided Verification}, pages 251--267. Springer International Publishing, 2015.

\bibitem{BriceRB21}
L{\'{e}}onard Brice, Jean{-}Fran{\c{c}}ois Raskin, and Marie van~den Bogaard.
\newblock Subgame-perfect equilibria in mean-payoff games.
\newblock In Serge Haddad and Daniele Varacca, editors, {\em 32nd International
  Conference on Concurrency Theory, {CONCUR} 2021, August 24-27, 2021, Virtual
  Conference}, volume 203 of {\em LIPIcs}, pages 8:1--8:17. Schloss Dagstuhl -
  Leibniz-Zentrum f{\"{u}}r Informatik, 2021.
\newblock \href {https://doi.org/10.4230/LIPIcs.CONCUR.2021.8}
  {\path{doi:10.4230/LIPIcs.CONCUR.2021.8}}.

\bibitem{BdPS13}
Thomas Brihaye, Julie De~Pril, and Sven Schewe.
\newblock Multiplayer cost games with simple nash equilibria.
\newblock In Sergei Artemov and Anil Nerode, editors, {\em Logical Foundations
  of Computer Science}, pages 59--73, Berlin, Heidelberg, 2013. Springer Berlin
  Heidelberg.

\bibitem{bulling2022combining}
Nils Bulling and Valentin Goranko.
\newblock Combining quantitative and qualitative reasoning in concurrent
  multi-player games.
\newblock {\em Autonomous Agents and Multi-Agent Systems}, 36:1--33, 2022.

\bibitem{camacho2019towards}
Alberto Camacho, Meghyn Bienvenu, and Sheila~A McIlraith.
\newblock Towards a unified view of ai planning and reactive synthesis.
\newblock In {\em Proceedings of the International Conference on Automated
  Planning and Scheduling}, volume~29, pages 58--67, 2019.

\bibitem{CCHKM05}
Arindam Chakrabarti, Krishnendu Chatterjee, Thomas~A. Henzinger, Orna
  Kupferman, and Rupak Majumdar.
\newblock Verifying quantitative properties using bound functions.
\newblock In Dominique Borrione and Wolfgang Paul, editors, {\em Correct
  Hardware Design and Verification Methods}, pages 50--64, Berlin, Heidelberg,
  2005. Springer Berlin Heidelberg.

\bibitem{CdAHS03}
Arindam Chakrabarti, Luca de~Alfaro, Thomas~A. Henzinger, and Mari{\"e}lle
  Stoelinga.
\newblock Resource interfaces.
\newblock In Rajeev Alur and Insup Lee, editors, {\em Embedded Software}, pages
  117--133, Berlin, Heidelberg, 2003. Springer Berlin Heidelberg.

\bibitem{chalkiadakis}
Georgios Chalkiadakis, Edith Elkind, and Michael~J. Wooldridge.
\newblock {\em Computational Aspects of Cooperative Game Theory}.
\newblock Synthesis Lectures on Artificial Intelligence and Machine Learning.
  Morgan {\&} Claypool Publishers, 2011.
\newblock \href {https://doi.org/10.2200/S00355ED1V01Y201107AIM016}
  {\path{doi:10.2200/S00355ED1V01Y201107AIM016}}.

\bibitem{Chatterjee2010}
Krishnendu Chatterjee, Laurent Doyen, Thomas~A. Henzinger, and Jean-François
  Raskin.
\newblock {Generalized Mean-payoff and Energy Games}.
\newblock In {\em FSTTCS}, pages 505--516, 2010.
\newblock \href {https://doi.org/10.4230/LIPIcs.FSTTCS.2010.505}
  {\path{doi:10.4230/LIPIcs.FSTTCS.2010.505}}.

\bibitem{4484790}
Krishnendu Chatterjee, Arkadeb Ghosal, Thomas~A. Henzinger, Daniel Iercan,
  Christoph~M. Kirsch, Claudio Pinello, and Alberto Sangiovanni-Vincentelli.
\newblock Logical reliability of interacting real-time tasks.
\newblock In {\em 2008 Design, Automation and Test in Europe}, pages 909--914,
  2008.
\newblock \href {https://doi.org/10.1109/DATE.2008.4484790}
  {\path{doi:10.1109/DATE.2008.4484790}}.

\bibitem{chatterjee2010strategy}
Krishnendu Chatterjee, Thomas~A Henzinger, and Nir Piterman.
\newblock Strategy logic.
\newblock {\em Information and Computation}, 208(6):677--693, 2010.

\bibitem{conitzer2023foundations}
Vincent Conitzer and Caspar Oesterheld.
\newblock Foundations of cooperative {AI}.
\newblock In Brian Williams, Yiling Chen, and Jennifer Neville, editors, {\em
  Thirty-Seventh {AAAI} Conference on Artificial Intelligence, {AAAI} 2023,
  Thirty-Fifth Conference on Innovative Applications of Artificial
  Intelligence, {IAAI} 2023, Thirteenth Symposium on Educational Advances in
  Artificial Intelligence, {EAAI} 2023, Washington, DC, USA, February 7-14,
  2023}, pages 15359--15367. {AAAI} Press, 2023.
\newblock URL: \url{https://doi.org/10.1609/aaai.v37i13.26791}, \href
  {https://doi.org/10.1609/AAAI.V37I13.26791}
  {\path{doi:10.1609/AAAI.V37I13.26791}}.

\bibitem{dafoe2021cooperative}
Allan Dafoe, Yoram Bachrach, Gillian Hadfield, Eric Horvitz, Kate Larson, and
  Thore Graepel.
\newblock Cooperative ai: machines must learn to find common ground.
\newblock {\em Nature}, 593(7857):33--36, 2021.

\bibitem{dafoe2020open}
Allan Dafoe, Edward Hughes, Yoram Bachrach, Tantum Collins, Kevin~R McKee,
  Joel~Z Leibo, Kate Larson, and Thore Graepel.
\newblock Open problems in cooperative ai.
\newblock {\em arXiv preprint arXiv:2012.08630}, 2020.

\bibitem{Degorre2010}
Aldric Degorre, Laurent Doyen, Raffaella Gentilini, Jean-François Raskin, and
  Szymon Toruńczyk.
\newblock Energy and mean-payoff games with imperfect information.
\newblock In Anuj Dawar and Helmut Veith, editors, {\em Computer Science
  Logic}, pages 260--274, Berlin, Heidelberg, 2010. Springer Berlin Heidelberg.

\bibitem{demri2002complexity}
St{\'e}phane Demri and Philippe Schnoebelen.
\newblock The complexity of propositional linear temporal logics in simple
  cases.
\newblock {\em Information and Computation}, 174(1):84--103, 2002.

\bibitem{EM79}
A.~Ehrenfeucht and J.~Mycielski.
\newblock Positional strategies for mean payoff games.
\newblock {\em Int. J. Game Theory}, 8(2):109–113, jun 1979.
\newblock \href {https://doi.org/10.1007/BF01768705}
  {\path{doi:10.1007/BF01768705}}.

\bibitem{EL87}
E.~Allen Emerson and Chin-Laung Lei.
\newblock Modalities for model checking: Branching time logic strikes back.
\newblock {\em Sci. Comput. Program.}, 8(3):275–306, jun 1987.
\newblock \href {https://doi.org/10.1016/0167-6423(87)90036-0}
  {\path{doi:10.1016/0167-6423(87)90036-0}}.

\bibitem{filippidis2016control}
Ioannis Filippidis, Sumanth Dathathri, Scott~C Livingston, Necmiye Ozay, and
  Richard~M Murray.
\newblock Control design for hybrid systems with tulip: The temporal logic
  planning toolbox.
\newblock In {\em 2016 IEEE Conference on Control Applications (CCA)}, pages
  1030--1041. IEEE, 2016.

\bibitem{fisman2010rational}
Dana Fisman, Orna Kupferman, and Yoad Lustig.
\newblock Rational synthesis.
\newblock In {\em Tools and Algorithms for the Construction and Analysis of
  Systems: 16th International Conference, TACAS 2010, Held as Part of the Joint
  European Conferences on Theory and Practice of Software, ETAPS 2010, Paphos,
  Cyprus, March 20-28, 2010. Proceedings 16}, pages 190--204. Springer, 2010.

\bibitem{gillies1959solutions}
Donald~B Gillies.
\newblock Solutions to general non-zero-sum games.
\newblock {\em Contributions to the Theory of Games}, 4(40):47--85, 1959.

\bibitem{grunbaum2003convex}
Branko Grünbaum, Volker Kaibel, Victor Klee, and Günter~M. Ziegler.
\newblock {\em Convex polytopes}.
\newblock Springer, New York, 2003.
\newblock URL:
  \url{http://www.springer.com/mathematics/geometry/book/978-0-387-00424-2}.

\bibitem{9812068}
Zhaoyuan Gu, Nathan Boyd, and Ye~Zhao.
\newblock Reactive locomotion decision-making and robust motion planning for
  real-time perturbation recovery.
\newblock In {\em 2022 International Conference on Robotics and Automation
  (ICRA)}, pages 1896--1902, 2022.
\newblock \href {https://doi.org/10.1109/ICRA46639.2022.9812068}
  {\path{doi:10.1109/ICRA46639.2022.9812068}}.

\bibitem{HLNW21}
Julian Gutierrez, Lewis Hammond, Anthony~W. Lin, Muhammad Najib, and Michael~J.
  Wooldridge.
\newblock Rational verification for probabilistic systems.
\newblock In Meghyn Bienvenu, Gerhard Lakemeyer, and Esra Erdem, editors, {\em
  Proceedings of the 18th International Conference on Principles of Knowledge
  Representation and Reasoning, {KR} 2021, Online event, November 3-12, 2021},
  pages 312--322, 2021.
\newblock \href {https://doi.org/10.24963/kr.2021/30}
  {\path{doi:10.24963/kr.2021/30}}.

\bibitem{gutierrez2015iterated}
Julian Gutierrez, Paul Harrenstein, and Michael Wooldridge.
\newblock Iterated boolean games.
\newblock {\em Inform. Comput.}, 242:53--79, June 2015.
\newblock \href {https://doi.org/10.1016/j.ic.2015.03.011}
  {\path{doi:10.1016/j.ic.2015.03.011}}.

\bibitem{GutierrezHW17}
Julian Gutierrez, Paul Harrenstein, and Michael~J. Wooldridge.
\newblock From model checking to equilibrium checking: Reactive modules for
  rational verification.
\newblock {\em Artif. Intell.}, 248:123--157, 2017.
\newblock \href {https://doi.org/10.1016/j.artint.2017.04.003}
  {\path{doi:10.1016/j.artint.2017.04.003}}.

\bibitem{gutierrez2023cooperative}
Julian Gutierrez, Szymon Kowara, Sarit Kraus, Thomas Steeples, and Michael
  Wooldridge.
\newblock Cooperative concurrent games.
\newblock {\em Artificial Intelligence}, 314:103806, 2023.

\bibitem{Gutierrez2019}
Julian Gutierrez, Sarit Kraus, and Michael~J. Wooldridge.
\newblock Cooperative concurrent games.
\newblock In Edith Elkind, Manuela Veloso, Noa Agmon, and Matthew~E. Taylor,
  editors, {\em Proceedings of the 18th International Conference on Autonomous
  Agents and MultiAgent Systems, \{AAMAS\} '19, Montreal, QC, Canada, May
  13-17, 2019}, pages 1198--1206. International Foundation for Autonomous
  Agents and Multiagent Systems, 2019.
\newblock URL: \url{http://dl.acm.org/citation.cfm?id=3331822}.

\bibitem{GNPW19}
Julian Gutierrez, Muhammad Najib, Giuseppe Perelli, and Michael Wooldridge.
\newblock {Equilibrium Design for Concurrent Games}.
\newblock In Wan Fokkink and Rob van Glabbeek, editors, {\em 30th International
  Conference on Concurrency Theory (CONCUR 2019)}, volume 140 of {\em Leibniz
  International Proceedings in Informatics (LIPIcs)}, pages 22:1--22:16,
  Dagstuhl, Germany, 2019. Schloss Dagstuhl--Leibniz-Zentrum fuer Informatik.
\newblock URL: \url{http://drops.dagstuhl.de/opus/volltexte/2019/10924}, \href
  {https://doi.org/10.4230/LIPIcs.CONCUR.2019.22}
  {\path{doi:10.4230/LIPIcs.CONCUR.2019.22}}.

\bibitem{gutierrez2019on}
Julian Gutierrez, Muhammad Najib, Giuseppe Perelli, and Michael Wooldridge.
\newblock On computational tractability for rational verification.
\newblock In {\em Proceedings of the Twenty-Eighth International Joint
  Conference on Artificial Intelligence}, pages 329--335. International Joint
  Conferences on Artificial Intelligence Organization, August 2019.
\newblock \href {https://doi.org/10.24963/ijcai.2019/47}
  {\path{doi:10.24963/ijcai.2019/47}}.

\bibitem{4567949}
Dexter Kozen.
\newblock Lower bounds for natural proof systems.
\newblock In {\em 18th Annual Symposium on Foundations of Computer Science
  (sfcs 1977)}, pages 254--266, 1977.
\newblock \href {https://doi.org/10.1109/SFCS.1977.16}
  {\path{doi:10.1109/SFCS.1977.16}}.

\bibitem{kupferman2016synthesis}
Orna Kupferman, Giuseppe Perelli, and Moshe~Y. Vardi.
\newblock Synthesis with rational environments.
\newblock {\em Ann. Math. Artif. Intel.}, 78(1):3--20, June 2016.
\newblock \href {https://doi.org/10.1007/s10472-016-9508-8}
  {\path{doi:10.1007/s10472-016-9508-8}}.

\bibitem{maoz2015gr}
Shahar Maoz and Jan~Oliver Ringert.
\newblock Gr (1) synthesis for ltl specification patterns.
\newblock In {\em Proceedings of the 2015 10th joint meeting on foundations of
  software engineering}, pages 96--106, 2015.

\bibitem{MS12}
Shahar Maoz and Yaniv Sa'ar.
\newblock Assume-guarantee scenarios: Semantics and synthesis.
\newblock In Robert~B. France, J{\"u}rgen Kazmeier, Ruth Breu, and Colin
  Atkinson, editors, {\em Model Driven Engineering Languages and Systems},
  pages 335--351, Berlin, Heidelberg, 2012. Springer Berlin Heidelberg.

\bibitem{markey2004past}
Nicolas Markey.
\newblock Past is for free: on the complexity of verifying linear temporal
  properties with past.
\newblock {\em Acta Informatica}, 40:431--458, 2004.

\bibitem{mogavero2014reasoning}
Fabio Mogavero, Aniello Murano, Giuseppe Perelli, and Moshe~Y Vardi.
\newblock Reasoning about strategies: On the model-checking problem.
\newblock {\em ACM Transactions on Computational Logic (TOCL)}, 15(4):1--47,
  2014.

\bibitem{papadimitriou1994computational}
C.~Papadimitriou.
\newblock {\em {Computational complexity}}.
\newblock Addison-Wesley, Reading, Massachusetts, 1994.

\bibitem{perelli2019enforcing}
G.~Perelli.
\newblock Enforcing equilibria in multi-agent systems.
\newblock In {\em Proceedings of the 18th International Conference on
  Autonomous Agents and MultiAgent Systems}, AAMAS '19, pages 188--196, 2019.
\newblock URL: \url{http://dl.acm.org/citation.cfm?id=3306127.3331692}.

\bibitem{Pnueli1977}
Amir Pnueli.
\newblock The temporal logic of programs.
\newblock In {\em 18th Annual Symposium on Foundations of Computer Science
  (sfcs 1977)}, pages 46--57. IEEE, September 1977.
\newblock \href {https://doi.org/10.1109/sfcs.1977.32}
  {\path{doi:10.1109/sfcs.1977.32}}.

\bibitem{ray1997equilibrium}
Debraj Ray and Rajiv Vohra.
\newblock Equilibrium binding agreements.
\newblock {\em Journal of Economic theory}, 73(1):30--78, 1997.

\bibitem{scarf1967core}
Herbert~E Scarf.
\newblock The core of an n person game.
\newblock {\em Econometrica: Journal of the Econometric Society}, pages 50--69,
  1967.

\bibitem{scarf1971existence}
Herbert~E Scarf.
\newblock On the existence of a coopertive solution for a general class of
  n-person games.
\newblock {\em Journal of Economic Theory}, 3(2):169--181, 1971.

\bibitem{Sistla1985}
A.~P. Sistla and E.~M. Clarke.
\newblock The complexity of propositional linear temporal logics.
\newblock {\em J. ACM}, 32(3):733--749, July 1985.
\newblock \href {https://doi.org/10.1145/3828.3837}
  {\path{doi:10.1145/3828.3837}}.

\bibitem{steeples2021mean}
Thomas Steeples, Julian Gutierrez, and Michael Wooldridge.
\newblock Mean-payoff games with $\omega$-regular specifications.
\newblock In {\em Proceedings of the 20th International Conference on
  Autonomous Agents and MultiAgent Systems}, pages 1272--1280, 2021.

\bibitem{ummels2011the}
M.~Ummels and D.~Wojtczak.
\newblock {The Complexity of {Nash} Equilibria in Limit-Average Games}.
\newblock In {\em CONCUR}, pages 482--496, 2011.
\newblock \href {https://doi.org/10.1007/978-3-642-23217-6{\textbackslash}\_32}
  {\path{doi:10.1007/978-3-642-23217-6{\textbackslash}\_32}}.

\bibitem{uyanik2015nonemptiness}
Metin Uyan{\i}k.
\newblock On the nonemptiness of the $\alpha$-core of discontinuous games:
  Transferable and nontransferable utilities.
\newblock {\em Journal of Economic Theory}, 158:213--231, 2015.

\bibitem{Velner2015}
Yaron Velner, Krishnendu Chatterjee, Laurent Doyen, Thomas~A. Henzinger,
  Alexander Rabinovich, and Jean-Fran\c{c}ois Raskin.
\newblock The complexity of multi-mean-payoff and multi-energy games.
\newblock {\em Inform. Comput.}, 241:177--196, April 2015.
\newblock \href {https://doi.org/10.1016/j.ic.2015.03.001}
  {\path{doi:10.1016/j.ic.2015.03.001}}.

\bibitem{wooldridge2016rational}
M.~Wooldridge, J.~Gutierrez, P.~Harrenstein, E.~Marchioni, G.~Perelli, and
  A.~Toumi.
\newblock {Rational Verification: {From} Model Checking to Equilibrium
  Checking}.
\newblock In {\em \{AAAI\}}, pages 4184--4191. \{AAAI\} Press, 2016.

\bibitem{zwick1996the}
Uri Zwick and Mike Paterson.
\newblock The complexity of mean payoff games on graphs.
\newblock {\em Theor. Comput. Sci.}, 158(1-2):343--359, May 1996.
\newblock \href {https://doi.org/10.1016/0304-3975(95)00188-3}
  {\path{doi:10.1016/0304-3975(95)00188-3}}.

\end{thebibliography}

\newpage

\appendix

\section{Appendix: Proofs}

\subsection{Proof of Proposition~\ref{prpn:core-not-pareto-optimal}}

\begin{figure}[ht]
	\centering
	\scalebox{0.7}{
\begin{tikzpicture}[state/.style={circle, draw, minimum size=1.2cm}, node distance=1.5cm]
		\node[state] (s^0) {\small $(0,0)$};
		\node[] (start) [left =1cm of s^0] {};
		\node[state] (s^2) [right = of s^0] {\small $(0,1)$};
		\node[state] (s^1) [above = of s^2] {\small $(0,2)$};
		\node[state] (s^3) [below = of s^2] {\small $(0,-1)$};

		\draw [-{Latex[width=2mm]}]
		(start) edge node{} (s^0)
		(s^0) edge[] node[right]{\small $(R,R)$} (s^1)
		(s^0) edge[] node[above,align=center]{\small $(L,L)$\\$(R,L)$} (s^2)
		(s^0) edge[] node[right]{\small $(L,R)$} (s^3)

		(s^1) edge[loop, out=20, in=340, distance=1cm] node[right]{$*$} (s^1)
		(s^2) edge[loop, out=20, in=340, distance=1cm] node[right]{$*$} (s^2)
		(s^3) edge[loop, out=20, in=340, distance=1cm] node[right]{$*$} (s^3)
		;
\end{tikzpicture}
}
	\caption{An example of a game with a member of the core that is not Pareto optimal. The symbol $ \ast $ is a wildcard that matches all possible actions.\label{fig:non-pareto-optimal-core}}
\end{figure}
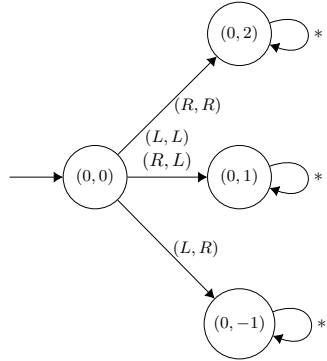

\PROPstrict*

\begin{proof}
	Let \(\Game\) be a mean-payoff game with two players, \(\Ag = \{1,2\}\), and four states. The game graph arena is shown in Figure~\ref{fig:non-pareto-optimal-core}.
	Observe that \((L,L)\) is in the core\footnote{Note that whilst this should be an infinite sequence of actions, only the actions in the first round matter. To avoid clutter, we omit the rest.}: player 1 has no incentive to deviate as they have a constant payoff, and so the coalitions \(\{1\}\) and \(\{1,2\}\) do not have beneficial deviations. Player 2 receives a payoff of \(1\) under \((L, L)\) and moving to \(R\) is not a beneficial deviation, as \((L, R)\) leads to a payoff of \(-1\). Thus, \((L,L)\) lies in the core. However, this strategy is not Pareto optimal, as it is (weakly) dominated by \((R,R)\).
\end{proof}

\subsection{Proof of Proposition~\ref{prpn:core-pareto-optimal}}

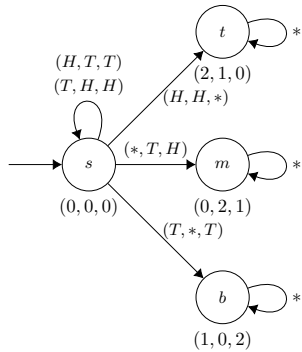
\begin{figure}[ht]
	\centering
	\scalebox{0.7}{
	\begin{tikzpicture}[state/.style={circle, draw, minimum size=1cm}, node distance=1.5cm]
		\node[state, label=below:{$ (0,0,0) $}] (s^0) {\small $s$};
		\node[] (start) [left =1cm of s^0] {};
		\node[state, label=below:{$ (0,2,1) $}] (s^2) [right = of s^0] {\small $m$};
		\node[state, label=below:{$ (2,1,0) $}] (s^1) [above = of s^2] {\small $t$};
		\node[state, label=below:{$ (1,0,2) $}] (s^3) [below = of s^2] {\small $b$};
		
		\draw [-{Latex[width=2mm]}]
		(start) edge node{} (s^0)
		(s^0) edge[] node[right]{\small $(H,H,\ast)$} (s^1)
		(s^0) edge[] node[above]{\small $(*,T,H)$} (s^2)
		(s^0) edge[] node[right]{\small $(T,\ast,T)$} (s^3)
		
		(s^1) edge[loop, out=20, in=340, distance=1cm] node[right]{$*$} (s^1)
		(s^2) edge[loop, out=20, in=340, distance=1cm] node[right]{$*$} (s^2)
		(s^3) edge[loop, out=20, in=340, distance=1cm] node[right]{$*$} (s^3)
		
		(s^0) edge[loop, out=70, in=110, distance=1cm] node[above, align=center]{\small $(H,T,T)$\\ \small $(T,H,H)$} (s^0)
		;
	\end{tikzpicture}
}
	\caption{An example of a game with an empty core. \label{fig:prop2}}
\end{figure}

\PROPPO*

\begin{proof}
	Let $ \Game $ be a mean-payoff game with $ \Ag = \{1,2,3\} $. The game arena is shown in Figure~\ref{fig:prop2}.
	The Pareto optimal set is $ \PO(G^{\Ag}) = \{(2,1,0),(0,2,1),(1,0,2)\} $. Observe that the game has empty core: if the players stay in $ s $ forever, then $ \{1,2\} $ can beneficially deviate to $ t $. If the play goes to $ t $, then $ \{2,3\} $ can beneficially deviate to $ m $. Similar arguments can be used for $ m $ and $ b $; thus, no (Pareto optimal) strategy profile lies in the core.
\end{proof}

\subsection{Proof of Theorem~\ref{thm:e-ben-dev}}

\THMbendev*

\begin{proof}
	To solve \ebendev, we reduce it to \dominated as follows. First we compute \( (\pay_i(\strpElm))_{i \in \Ag} \) in \pspace (Lemma~\ref{lem:payoff}). Then we can query whether \( (\Game,\sinit,(\pay_i(\strpElm))_{i \in \Ag}) \in \) \dominated. Since $ \SigmaPTwo \subseteq \pspace $, \ebendev can be solved in \pspace.
	
	For the lower bound, we reduce from the non-emptiness problem of the intersection of deterministic finite automata (DFA) that is known to be \pspace-hard~\cite{4567949}. Let $ A_1,\dots,A_n $ be a set of deterministic finite automata (DFAs), and let $ F_i = \{ q_i^* \} $ be the set of accepting state of $ A_i $. Note that we can always assume that $F_i$ only has one state; otherwise, we can simply introduce a new symbol in the alphabet (call it $a$), a new state $f_i$ for $A_i$, and define the final state of $A_i$ to be $f_i$, as well as defining $\Delta_i(q,a) := f_i$, for each $q \in F_i$, where $\Delta_i$ is the transition function of $F_i$.
	We construct from each $ A_i = (Q_i,\Sigma_i,\delta_i,q_i^0,F_i) $ a strategy $ \strElm_{i} = (Q_i,q_i^0,\delta_i,\tau_i) $ where $ \tau_i(q_i) = q_i $. We build a game with $ \Ag = \{1,\dots,n\} $ and arena with 3 states $ \St = \{ s_0,s_1,s_2 \} $. For each $ i \in \Ag, \Ac_i = Q_i \cup \{d_i\} $, where $ d_i $ is a fresh variable. The transition function is defined as follows:
	
	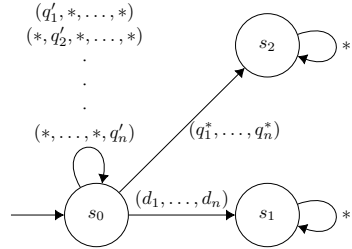
\begin{figure}[H]
		\centering
		\scalebox{0.7}{
		\begin{tikzpicture}[state/.style={circle, draw, minimum size=1.2cm}, node distance=2cm]
			\node[state] (s^0) {$s_0$};
			\node[] (start) [left =1cm of s^0] {};
			\node[state] (s^2) [right = of s^0] {$s_1$};
			\node[state] (s^1) [above = of s^2] {$s_2$};
			
			\draw [-{Latex[width=2mm]}]
			(start) edge node{} (s^0)
			(s^0) edge[] node[right]{$(q_1^*,\dots,q_n^*)$} (s^1)
			(s^0) edge[] node[above,align=center]{$(d_1,\dots,d_n)$} (s^2)
			
			(s^0) edge[loop, out=120, in=80, distance=1cm] node[above,align=center]{$(q_1',\ast,\dots,\ast)$\\$(\ast,q_2',\ast,\dots,\ast)$\\$ \cdot $\\$ \cdot $\\$ \cdot $\\$(\ast,\dots,\ast,q_n')$} (s^0)
			(s^1) edge[loop, out=20, in=340, distance=1cm] node[right]{$*$} (s^1)
			(s^2) edge[loop, out=20, in=340, distance=1cm] node[right]{$*$} (s^2)
			;
		\end{tikzpicture}
		}
		\caption{The game arena where $ q_i' \neq q_i^* $.}
		\label{fig:bendev}
	\end{figure}

	The weight function is given as follows.
	
	\begin{center}
		\begin{tabular}{c c}
			\toprule
			\(s \in \St\)    & $ (\wFun_{i}(s))_{i \in \Ag} $ \\
			\midrule
			$ s_0 $  &     $ (0,\dots,0) $         \\
			$ s_1 $  &     $ (1,\dots,1) $         \\
			$ s_2 $  &     $ (1,\dots,1) $         \\
			\bottomrule
		\end{tabular}
	\end{center}

	Given $ (\Game,\strpElm) $ where $ \strpElm = (\strElm_1,\dots,\strElm_n) $, we claim that $ (\Game,\strpElm) \notin $ \ebendev if and only if the intersection of $ A_1,\dots,A_n $ has non-empty language. From left to right: it is easy to see that in order for $ \strpElm $ to admit no beneficial deviation, the game has to eventually enter $ s_2 $, because otherwise the grand coalition can deviate to $ s_1 $ and obtain better payoffs. The only possible way to enter $ s_2 $ is when each of $ A_i $ arrives at the accepting state, and thus the intersection has non-empty language. From right to left, we argue in a similar way.
\end{proof}

\subsection{Proof of Theorem~\ref{thm:dominated}}\label{app:dominated}

\THMdominated*

\begin{proof}
	Observe that an instance $ (\Game,s,\vec{x}) \in $ \textsc{Dominated} has a witness vector $ (x_i')_{i \in C} $ that lies in the intersection of a polyhedron $ P^C \in \polyset(G^C,s) $ and the set $ \{ \vec{y} \in \Real^c \mid \forall i \in C: y_i \geq x_i \} $. Such an intersection forms a polyhedron $ \poly(\lambda) $, definable by a system of linear inequalities $ \lambda $. By Lemma~\ref{lem:representation}, each $ (\vec{a},b) \in \lambda $ has polynomial representation in the size of $ G^C $. Therefore, $ (x_i')_{i \in C} $ has a representation that is polynomial in the size of $ \Game $. To solve \textsc{Dominated}, we provide Algorithm~\ref{alg:dominated}. The correctness follows directly from the definition of \textsc{Dominated}. For the upper bound: since $ (x_i')_{i \in C} $ is of polynomial size, line 1 can be done in \( \np \). In line 2, we have subprocedure \textsc{Sequentialise} that builds and returns sequentialisation of $ \Game $ w.r.t. coalition $ C $; this can be done in polynomial time. Finally, line 3 is in \( \conp \)~\cite[Theorem 3, Lemma 6]{Velner2015}. Therefore, the algorithm runs in \( \np^{\conp} = \SigmaPTwo \).
	
	For the lower bound, we reduce from \( \QSAT_2 (3\DNF) \) (satisfiability of quantified Boolean formulae with 2 alternations and 3DNF clauses), which is known to be \( \SigmaPTwo \)-hard~\cite{papadimitriou1994computational}. Consider a formula of the form
	\[ \Phi \coloneqq \exists x_1 \cdots \exists x_p \forall y_1 \cdots \forall y_q C_1 \vee \cdots \vee C_r \]
	where each \( C_i \) is the conjunction of three literals \( C_i = l_{i,1} \wedge l_{i,2} \wedge l_{i,3} \), and the literals are of the form \( x_k, \neg x_k, y_k, \) or \( \neg y_k \). For clauses \( C \) and \( C' \), we say that they are not \textit{clashing} if there is no literal \( x_k \) appears in \( C \) and \( \neg x_k \) in \( C' \).
	
	For a given formula \( \Phi \) we build a corresponding game \( \Game^{\Phi} \) such that \( (\Game^{\Phi},\sinit,(-1,\dots,-1,0)) \in \textsc{Dominated} \) if and only if \( \Phi \) is satisfiable, as follows.
	\begin{itemize}
		\item \( \Ag = \{1,\dots,2q,E,A\} \);
		\item \( \St = \{ \sinit, C_1,\dots,C_r, l_{1,1}, \dots, l_{r,3}, \sink \} \) where
		\begin{itemize}
			\item the states $\sinit $ and $x$-literal states are controlled by player $ E $,
			\item each state $ l_{i,j} $ of the from $y_k$ (resp. $ \neg y_k $) is controlled by player $2k$ (resp. $2k-1$), and
			\item \( \{{C_1},\dots,{C_r}\} \) (i.e., the clause states) by player $ A $;
		\end{itemize}
		
		\item the transition function is given as:
		\begin{itemize}

			\item from \( \sinit \), player $ E $ can decide to which state in \( \{{C_1},\dots,{C_r}\} \) the play will proceed---she picks the clause;
			
			\item from each state \( {C_i} \), player $ A $ can decide to which state in \( \{{l_{i,1}},\dots, {l_{i,3}}\} \) the play will proceed---he picks the literal;
			
			\item from each \( {l_{i,j}} \), there is a self-loop transition, 
			
			\item from each \( {l_{i,j}} \) of the form \( y_k \) (resp. \( \neg y_k \)), the transitions are controlled by player $2k$ (resp. $2k-1$), and defined as follows:
			\begin{itemize}
				\item there is a transition from \( {l_{i,j}} \) to every \( {C_h}, i \neq h \), where \( y_k \) or \( \neg y_k \) occurs in \( C_h \), and \( C_i, C_h\) are not clashing, and
				\item there is also a transition to $ \sink $.
			\end{itemize}
		
			\item $ \sink $ has only self-loop transition.
			
		\end{itemize}
		
		\item the weight function is given as:
		\begin{itemize}
			\item for a literal state \( {l_{i,j}} \)
			\begin{itemize}
				\item if \( l_{i,j} \) is of the form \( y_k \), then \( \wFun_{2k-1}({l_{i,j}}) = 2q \) and \( \wFun_{2k}({l_{i,j}}) = -2q \), and for each \( a \in \Ag \setminus \{2k-1,2k\} \), \( \wFun_{a}({l_{i,j}}) = 0 \);
				
				\item if \( l_{i,j} \) is of the form \( \neg y_k \), then \( \wFun_{2k-1}({l_{i,j}}) = -2q \) and \( \wFun_{2k}({l_{i,j}}) = 2q \) and for each \( a \in \Ag \setminus \{2k-1,2k\} \), \( \wFun_{a}({l_{i,j}}) = 0 \);
				
				\item if \( l_{i,j} \) is of the form \( x_k \) or \( \neg x_k \), \( (\wFun_{a}({l_{i,j}}))_{a \in \Ag} = \vec{0} \).
			\end{itemize}
			
			\item for each non-literal state \( s \in \{ \sinit,{C_1},\dots,{C_r} \} \), we have \( (\wFun_{i}(s))_{i \in \Ag} = \vec{0} \).
			
			\item for each $ i \in \Ag \setminus \{A\}, \wFun_{i}(\sink) = -1 $ and $ \wFun_{A}(\sink) = 0 $.
		\end{itemize}
	\end{itemize}
	
	\begin{figure*}[ht]
		\centering
		\scalebox{0.7}{
\begin{tikzpicture}[state/.style={circle, draw, minimum size=0.9cm}, dbl/.style={circle, draw, fill=gray, minimum size=0.9cm}, sqr/.style={rectangle, draw, minimum size=0.9cm},
	dsqr/.style={fill=gray, rectangle, draw, minimum size=0.9cm},
	bezier bounding box=true]
		\node[state] (s_init) {$s_{\text{init}}$};
		\coordinate [above =0.5cm of s_init] (start);
                
		\node[sqr] (C_2) [below = of s_init] {$C_2$};
		\node[sqr] (C_1) [left = 4.2cm of C_2] {$C_1$};
		\node[sqr] (C_3) [right = 4.2cm of C_2] {$C_3$};

                \node[state] (C_1-x_2) [below = of C_1] {$x_2$};
                \node[state] (C_1-x_1) [left = 0.5cm of C_1-x_2] {$x_1$};
                \node[dbl] (C_1-x_3) [right = 0.5cm of C_1-x_2] {$y_1$};

                \node[state] (C_2-x_2) [below = of C_2] {$\lnot x_2$};
                \node[state] (C_2-x_1) [left = 0.5cm of C_2-x_2] {$x_1$};
                \node[dbl] (C_2-x_3) [right = 0.5cm of C_2-x_2] {$\lnot y_2$};

                \node[state] (C_3-x_2) [below = of C_3] {$x_1$};
                \node[state] (C_3-x_1) [left = 0.5cm of C_3-x_2] {$x_2$};
                \node[dbl] (C_3-x_3) [right = 0.5cm of C_3-x_2] {$\lnot y_1$};

                \coordinate [below = of C_2-x_2] (C_1-x_3-curve-stop-1);
                \coordinate [below = of C_3-x_3] (C_1-x_3-curve-stop-2) ;
                \coordinate [right = 1.5cm of C_3-x_3] (C_1-x_3-curve-stop-3);

                \coordinate [below right = of C_2-x_3] (C_3-x_3-curve-stop-1);
                \coordinate [below right =of C_2-x_1] (C_3-x_3-curve-stop-2);

		\draw[-{Latex[width=2mm]}] (start) --  (s_init);
		\draw[-{Latex[width=2mm]}] (s_init) --  (C_1);
		\draw[-{Latex[width=2mm]}] (s_init) --  (C_2);
		\draw[-{Latex[width=2mm]}] (s_init) --  (C_3);

                \draw[-{Latex[width=2mm]}] (C_1) -- (C_1-x_1);
                \draw[-{Latex[width=2mm]}] (C_1) -- (C_1-x_2);
                \draw[-{Latex[width=2mm]}] (C_1) -- (C_1-x_3);

                \draw[-{Latex[width=2mm]}] (C_2) -- (C_2-x_1);
                \draw[-{Latex[width=2mm]}] (C_2) -- (C_2-x_2);
                \draw[-{Latex[width=2mm]}] (C_2) -- (C_2-x_3);

                \draw[-{Latex[width=2mm]}] (C_3) --  (C_3-x_1);
                \draw[-{Latex[width=2mm]}] (C_3) -- (C_3-x_2);
                \draw[-{Latex[width=2mm]}] (C_3) --  (C_3-x_3);

		\draw[-{Latex[width=2mm]}] (C_1-x_1) edge[loop, out=290, in=250, distance=0.5cm] (C_1-x_1);
		\draw[-{Latex[width=2mm]}] (C_1-x_2) edge[loop, out=290, in=250, distance=0.5cm] (C_1-x_2);
		\draw[-{Latex[width=2mm]}] (C_1-x_3) edge[loop, out=290, in=250, distance=0.5cm] node[below]{$(4, -4,0,0,0,0)$} (C_1-x_3);

                \draw[-{Latex[width=2mm]}] (C_2-x_1) edge[loop, out=290, in=250, distance=0.5cm] (C_2-x_1);
		\draw[-{Latex[width=2mm]}] (C_2-x_2) edge[loop, out=290, in=250, distance=0.5cm] (C_2-x_2);
		\draw[-{Latex[width=2mm]}] (C_2-x_3) edge[loop, out=290, in=250, distance=0.5cm] node[below]{$(0,0,-4,4,0,0)$} (C_2-x_3);

                \draw[-{Latex[width=2mm]}] (C_3-x_1) edge[loop, out=290, in=250, distance=0.5cm] (C_3-x_1);
		\draw[-{Latex[width=2mm]}] (C_3-x_2) edge[loop, out=290, in=250, distance=0.5cm] (C_3-x_2);
		\draw[-{Latex[width=2mm]}] (C_3-x_3) edge[loop, out=290, in=250, distance=0.5cm] node[below]{$(-4,4,0,0,0,0)$} (C_3-x_3);

                \draw[-{Latex[width=2mm]}] (C_1-x_3-curve-stop-3) to [out=90, in=0] (C_3);
                \draw[-{Latex[width=2mm]}] (C_3-x_3-curve-stop-2) to [out=180, in=340] (C_1);

                \draw[] (C_1-x_3) to [out=340,in=180] (C_1-x_3-curve-stop-1);
                \draw[] (C_1-x_3-curve-stop-1) to [out=0,in=180] (C_1-x_3-curve-stop-2);
                \draw[] (C_1-x_3-curve-stop-2) to [out=0,in=270] (C_1-x_3-curve-stop-3);

                \draw[] (C_3-x_3) to [out=220,in=0] (C_3-x_3-curve-stop-1);
                \draw[] (C_3-x_3-curve-stop-1) to [out=180,in=0] (C_3-x_3-curve-stop-2);
		;
\end{tikzpicture}
}
		\caption{
			The game arena of $ \Game^\Phi $. White circle states are controlled by $ E $, square by $ A $, grey circles $ y_1, \neg y_1, \neg y_2 $ by players $ 1, 2, 4 $, respectively. The weight function is given as vectors shown below the states, and states without vectors have $ \vec{0} $. Furthermore, each grey circle state also has a transition to $ \sink $.
		}
		\label{fig:qsat-ap}
	\end{figure*}
	
	To illustrate the reduction, consider the formula \[ \Phi = \exists x_1 \exists x_2 \forall y_1 \forall y_2 (x_1 \wedge x_2 \wedge y_1) \vee (x_1 \wedge \neg x_2 \wedge \neg y_2) \vee (x_1 \wedge x_2 \wedge \neg y_1). \]
	We build a corresponding game $ \Game^\Phi $ such that $ (\Game^\Phi, \sinit, (-1,-1,-1,-1,-1,0)) = \chi \in $ \textsc{Dominated} if and only if $ \Phi $ is satisfiable.
	To this end, we construct the game \( \Game^\Phi \) in Figure~\ref{fig:qsat-ap} with $ \Ag = \{1,2,3,4,E,A\} $ and the weight function given as vectors, such that for a given vector $ (w_1,...,w_6) $ in state $ s $, $ \wFun_{i}(s) = w_i, i \in \{1,2,3,4\} $ and $ \wFun_{E}(s) = w_5, \wFun_{A}(s) = w_6$. The $ \sink $ only has transition to itself and its weights is given by the vector $ (-1,-1,-1,-1,-1,0) $. The intuition is that if $ \Phi $ is satisfiable, then there is a joint strategy $ \strpElm_{C} $ by $ C = \Ag \setminus \{A\} $ that guarantees a payoff of $ 0 $ for each $ i \in C $. If $ \Phi $ is not satisfiable, then $ A $ has a strategy that visits some state $ y_k $ (resp. $ \neg y_k $) infinitely often and player $ 2k-1 $ (resp. $ 2k $) gets payoff $ < -1 $. Since $ y_k $ (resp. $ \neg y_k $) is controlled by $ 2k-1 $ (resp. $ 2k $), then the player will deviate to $ \sink $, and $ \chi \notin \dominated $. On the other hand, if $ \chi \in \dominated $, then there exists a strategy $ \strpElm_{C} $ which guarantees that the play: (a) ends up in some state $ x_k $ or $ \neg x_k $, or (b) visits both $ y_k $ and $ \neg y_k $ infinitely often. For the former, it means that there is a clause with only $ x $-literals, and the latter implies that for all (valid) assignments of $ y $-literals, there is an assignment for $ x $-literal that makes at least one clause evaluate to true. Both cases show that $ \Phi $ is satisfiable.
	Now, notice that the formula \( \Phi \) is satisfiable: take the assignment that set \( x_1 \) and \( x_2 \) to be both true. Indeed, $ \chi \in \textsc{Dominated} $: the coalition $\{1,2,3,4,E\}$ have a strategy that results in payoff vector \( \vec{0} \), e.g., take a strategy profile that corresponds to the cycle \( ({C_1}{y_1}{C_3}{\neg y_1})^{\omega} \).
	
	Observe that the construction above produces a game whose size is polynomial in the size of \( \Phi \). The numbers of players and states are clearly polynomial. The transition function has a polynomial representation. Checking for clashing clauses and determining transitions from literal states to clause states can be done in quadratic time. Overall, the construction of \( \Game^{\Phi} \) can be done in polynomial time.
	
	We show that \( (\Game^{\Phi},\sinit,(-1,\dots,-1,0)) \in \)  \textsc{Dominated} if and only if the formula \( \Phi \) is satisfiable.
	
	\( (\Leftarrow) \) Assume that \( \Phi \) is satisfiable, then there is a (partial) assignment \( v(x_1,\dots,x_p) \) such that the formula \( \forall y_1 \cdots \forall y_q C_1 \vee \cdots \vee C_r \) is valid. Let \( \strK \) and \( \strElm_A \) denote strategies of coalition \( \mathcal{K} = \Ag \setminus \{A\} \) and player \( A \), respectively. According to \cite{Velner2015}, it is enough to only consider memoryless strategies \( \strElm_A \). The strategies correspond to some assignments of variables, that is, by choosing the literal \( y_k \) or $\neg y_k$, player \( A \) sets the assignment of the literal such that it evaluates to false. Similarly, by choosing the clause \( C_i \), \( \coal \) pick the correct assignments for literals \( x_k \) or \( \neg x_k \) in \( C_i \). We distinguish between strategies that are \textit{admissible} and those that are not. A non-admissible strategy is a strategy that chooses two contradictory literals \( y_k \) in \( C \) and \( \neg y_k \) in \( C' \). If \( \strElm_A \) is non-admissible, then \( \mathcal{K} \) can achieve \( \vec{0} \) by choosing the strategy that alternates between \( C \) and \( C' \), and thus we have a yes-instance of \textsc{Dominated}.
	
	Now suppose that \( A \) chooses an admissible strategy \( \strElm_A \). Then it corresponds to a valid assignment \( v(y_1,\dots,y_q) \). Since for \( v(x_1,\dots,x_p) \) the formula \( \forall y_1 \cdots \forall y_q C_1 \vee \cdots \vee C_r \) is valid, the (full) assignment \( v(x_1,\dots,x_p,y_1,\dots,y_q) \) makes the formula \(  C_1 \vee \cdots \vee C_r \) evaluate to true. Thus, \( \coal \) can pick a clause state \( {C_i} \) that is true under \( v(x_1,\dots,x_p,y_1,\dots,y_q) \) and \( A \) picks a literal state of the form \( x_k \) or \( \neg x_k \) in clause \( C_i \), and not \( y_k \) or \( \neg y_k \) since it will contradict the assumption that \( C_i \) evaluates to true. Therefore, the strategy profile \( (\strK,\strElm_A) \) induces the payoff \( \vec{0} \), and we have a yes-instance of \textsc{Dominated}.
	
	\( (\Rightarrow) \) Assume that the strategy profile \( (\strK,\strElm_A) \) induces a payoff \( \pay_j((\strK,\strElm_A)) > -1 \) for each \( j \in \coal \). Let \( \mathcal{C} \) and \( \bar{\mathcal{C}} \) be the set of clauses that are chosen and not chosen in \( (\strK,\strElm_A) \), respectively. We define the (partial) assignment of \(v( x_1,\dots,x_p) \) as follows:
	\begin{enumerate}
		\item for each \( C_i \in \mathcal{C} \) and for each literal \( x_k \) or \( \neg x_k \) in \( C_i \)
		\begin{enumerate}
			\item \( v(x_k) \) is true;
			\item \( v(\neg x_k) \) is false;
		\end{enumerate}
		
		\item for each \( C_h \in \bar{\mathcal{C}} \) and for each literal \( x_k \) or \( \neg x_k \) in \( C_h \), if it does not appear in \( C_i \in \mathcal{C} \), then \( v(x_k) \) or \( v(\neg x_k) \) is true.
	\end{enumerate}
	
	Let \( v' \) be an (extended) arbitrary assignment of \( x_1,\dots,x_p,y_1,\dots,y_q \) compatible with \(v( x_1,\dots,x_p) \). Assume towards a contradiction that \( v' \) does not make any of the clauses evaluate to true. Then in each \( C_i \in \mathcal{C} \), \( \player \) can choose a literal that makes \( C_i \) false. Either (i) \( \player \) chooses a literal \( y_k \) or \( \neg y_k \) and there is only a self-loop from the sate \( {y_k} \) or \( {\neg y_k} \), or (ii) we visit some clauses infinitely often. We distinguish between these two cases:
	\begin{enumerate}
		\item[(i)] If the run arrives in literal \( y_k \) or \( \neg y_k \) and there is only a self-loop from the sate \( {y_k} \) or \( {\neg y_k} \), 
		then player $2k$ or $2k-1$ will choose to move into the sink state and the players get payoff $(-1,\dots,-1,0)$.
		This contradicts our previous assumption that \( \pay_j((\strK,\strP)) > -1 \) for each \( j \in \coal \);
		
		\item[(ii)] If the play visits some clauses infinitely often, then by the construction of the game graph there exists a literal state \( {y_k} \) (resp. \( {\neg y_k} \)) visited infinitely often with \( \wFun_{2k-1}({y_k}) = -2q \) (resp. \( \wFun_{2k}({\neg y_k}) = -2q \)) and the state \( {\neg y_k} \) (resp. \( {y_k} \)) is never visited. This means that either \( \pay_{2k}((\strK,\strP)) < -1 \) or \( \pay_{2k-1}((\strK,\strP)) < -1 \), and player $ 2k $ or $ 2k-1 $ will choose to go to $ \sink $ and the players get $ (-1,\dots,-1,0) $. This contradicts our previous assumption that \( \pay_j((\strK,\strP)) > -1 \) for each \( j \in \coal \);
	\end{enumerate}
	
	This implies that assignment \( v' \) makes at least one clause evaluate to true. Furthermore, since this holds for any arbitrary \( v' \) compatible with \(v( x_1,\dots,x_p) \), we conclude that \( \Phi \in \QSAT_2 \).
\end{proof}

\subsection{Proof of Theorem~\ref{thm:non-emptiness}}

\THMemptiness*

\begin{proof}
	To solve \nonemptiness, it is important to recall the following two results. Firstly, if a game $ \Game $ has a non-empty core, then there is a payoff vector $ \vec{x} $ resulting from $ \strpElm \in \core(\Game) $ whose representation is polynomial (Theorem~\ref{thm:polywitness}). Secondly, if $ \vec{x} $ is a witness for the core, then \( (\Game,\sinit,\vec{x}) \notin \) \dominated.
	With these observations, solving \textsc{Non-Emptiness} can be done by Algorithm~\ref{alg:nonemptiness}.
	The subprocedure in line 1 is polynomial. Line 2 is in \np (Theorem~\ref{thm:polywitness}) and we call $ \SigmaPTwo $ oracle for line 3. Thus, Algorithm~\ref{alg:nonemptiness} runs in $ \SigmaPThree $.
	
	For hardness, we reduce from \( \QSAT_3 (3\CNF) \) (satisfiability of quantified Boolean formulae with 3 alternations and 3CNF clauses). 
	Consider a formula of the form
	
	\[ \Psi \coloneqq \exists x_1 \cdots \exists x_p \forall y_1 \cdots \forall y_q \exists z_1 \cdots \exists z_t C_1 \wedge \cdots \wedge C_r. \]
	
	where each \( C_i \) is the disjunction of three literals \( C_i = l_{i,1} \vee l_{i,2} \vee l_{i,3} \), and the literals are of the form \( x_k, \neg x_k, y_k,\neg y_k, z_k,\) or \( \neg z_k \). For clauses \( C \) and \( C' \), we say that they are not $ y$\textit{-clashing} if there is no literal \( y_k \) (resp. $ \neg y_k $) appears in \( C \) and \( \neg y_k \) (resp. $ y_k $) in \( C' \).
	
	For a given formula $ \Psi $ we build a corresponding game $ \Game^\Psi $ such that the core of $ \Game^\Psi $ is not empty if and only if $ \Psi $ is satisfiable, as follows.
	
	\begin{itemize}
		\item $ \Ag = \{1,\dots,2p, 2p+1,\dots,2p+2t,E,A,P,Q,R \} $
		
		\item $ \St = \{ \sinit, \sink \} \cup \{C_v \vert 1 \leq v \leq r\} \cup \{l_{1,1},\dots,l_{r,3}\} $, where
		\begin{itemize}
			\item state $ \sinit $ is controlled by player $ A $
			
			\item states $ C_1, \dots, C_r $ are controlled by player $ E $
			
			\item each state $ l_{i,j} $ of the form $ x_k $ (resp.$ \neg x_k $) is controlled by player $ 2k-1 $ (resp. $ 2k $)
			
			\item each state $ l_{i,j} $ of the form $ z_k $ (resp.$ \neg z_k $) is controlled by player $ 2(p+k)-1 $ (resp. $ 2(p+k) $) and player $ A $, where player $ 2(p+k)-1 $/$ 2(p+k) $ has a ``veto'' power to either follow player $ A $'s decision or, instead, unilaterally choose to go to $ \sink $
			
			\item each state $ l_{i,j} $ of the form $ y_k $ or $ \neg y_k $ is controlled by player $ A $ \footnote{Note that the controller of these states is ultimately not important because, as later defined, from these states we can only go to $ \sink $.}
			
			\item the state $ \sink $ is a sink state, and implemented by a gadget that will be explained later.
			
		\end{itemize}
		
		\item the transition function is given as:
		\begin{itemize}
			
			\item from $ \sinit $ player $ A $ can choose to move to a clause state $ C_v, 1 \leq v \leq r $

			\item from a state $ C_v $ player $ E $ can choose to move to a literal state $ l_{v,j} $
			
			\item from a literal state $ l_{i,j} $ of the form $ x_k $ (resp. $ \neg x_k $), player $ 2k-1 $ (resp. $ 2k $) can choose to move to $ \sinit $ or $ \sink $
			
			\item from a literal state $ l_{i,j} $ of the form $ z_k $ or $ \neg z_k $, player $ E $ can choose to stay in the current state or to move to any clause state $ C' $ that is not \yclash with $ C_i $.
			
			\item from a literal state $ l_{i,j} $ of the form $ z_k $ (resp. $ \neg z_k $) player $ 2(p+k)-1 $ (resp. $ 2(p+k) $) can overrule player $ A $'s decision, and move to $ \sink $.
		\end{itemize}
		
		\item the weight function is given as:
		\begin{itemize}
			
			\item for each literal state $ l_{i,j} $
			\begin{itemize}
				\item if it is of the form $ x_k $, then $ \wFun_{2k-1}(l_{i,j}) = 3r, \wFun_{2k}(l_{i,j}) = -3r $ and for each $ a \in \Ag \setminus \{2k-1,2k\}, \wFun_a(l_{i,j}) = 0 $
				
				\item if it is of the form $ \neg x_k $, then $ \wFun_{2k}(l_{i,j}) = 3r, \wFun_{2k-1}(l_{i,j}) = -3r $ and for each $ a \in \Ag \setminus \{2k-1,2k\}, \wFun_a(l_{i,j}) = 0 $
				
				\item if it is of the form $ z_k $, then $ \wFun_{2(p+k)-1}(l_{i,j}) = 3r, \wFun_{2(p+k)}(l_{i,j}) = -3r $ and for each $ a \in \Ag \setminus \{2(p+k)-1,2(p+k)\}, \wFun_a(l_{i,j}) = 0 $
				
				\item if it is of the form $ \neg z_k $, then $ \wFun_{2(p+k)}(l_{i,j}) = 3r, \wFun_{2(p+k)-1}(l_{i,j}) = -3r $ and for each $ a \in \Ag \setminus \{2(p+k)-1,2(p+k)\}, \wFun_a(l_{i,j}) = 0 $
				
				\item otherwise, $ \wFun_{a}(l_{i,j}) = 0 $ for each $ a \in \Ag $.
			\end{itemize}
			
			\item $ \wFun_{a}(\sinit) = \wFun_{a}(s_\forall) = \wFun_{a}(C_i) = 0 $ for each $ a \in \Ag $ and $ 1 \leq i \leq r $.
		\end{itemize}
	\end{itemize}
	
	Now we explain the construction of $ \sink $ gadget which is a small variation of a game with an empty core provided in the proof of Proposition~\ref{prpn:core-pareto-optimal}. Consider a graph arena with four states $ I, U, M, B $ in which the players $ P, Q, R $ each has two actions: $ H, T $, and only the actions of those players matter in these states (i.e., the rest of the players are dummy players.) The weight function is given as follows:
	
	\begin{center}
		\begin{tabular}{CCCCCC}
			\hline
			\wFun_{a}(s) & P & Q & R & E & a \in \Ag \setminus \{P,Q,R,E\} \\
			\hline
			I & -1 & -1 & -1 & 0 & 1\\
			U & 2 & 1 & 0 & 0 & 1\\
			M & 0 & 2 & 1 & 0 & 1\\
			B & 1 & 0 & 2 & 0 & 1\\
			\hline
		\end{tabular}
	\end{center}
	
	The transition function is given below--we only specify the transitions for the state $ I $ as the other states only have self-loops.
	
	\begin{center}
		\begin{tabular}{CC}
			\hline
			(a_P,a_Q,a_R) & \St \\
			\hline
			(H,H,H) & U \\
			(H,H,T) & U \\
			
			(H,T,H) & M \\
			
			(H,T,T) & I \\
			(T,H,H) & I \\
			
			(T,H,T) & B \\
			
			(T,T,H) & M \\
			
			(T,T,T) & B \\
			\hline
		\end{tabular}
	\end{center}
	
	Observe that once we enter $ \sink $, we cannot get out. Furthermore, every strategy profile that starts at state $ I $ admits beneficial deviations. If the run stays at $ I $ forever, the players can beneficially deviate by moving to one of $ U,M,B $. However, if the game ends up at either of those states, then there will always be a coalition (of 2 players) that can beneficially deviate.
	
	To illustrate the construction, consider the formula
	\[ \Psi = \exists x_1 \exists x_2 \forall y_1 \exists z_1 (x_1 \vee x_2 \vee y_1) \wedge (\neg x_1 \vee y_1 \vee z_1) \wedge (\neg x_2 \vee \neg y_1 \vee \neg z_1) \]
	
	\begin{figure*}[ht]
		\centering
		\scalebox{0.7}{
\begin{tikzpicture}[state/.style={circle, draw, minimum size=0.9cm}, dbl/.style={circle, draw, fill=gray, minimum size=0.9cm}, sqr/.style={rectangle, draw, minimum size=0.9cm},
dsqr/.style={fill=gray, rectangle, draw, minimum size=0.9cm},
	bezier bounding box=true]
		\node[state] (s_init) {$s_{\text{init}}$};
		\coordinate [above =0.5cm of s_init] (start);
                
		\node[sqr] (C_2) [below = of s_init] {$C_2$};
		\node[sqr] (C_1) [left = 4.2cm of C_2] {$C_1$};
		\node[sqr] (C_3) [right = 4.2cm of C_2] {$C_3$};

                \node[dsqr] (C_1-x_2) [below = of C_1] {$x_2$};
                \node[dsqr, label=below:{\small $ (9,-9,0,0,0,0) $}] (C_1-x_1) [left = 0.5cm of C_1-x_2]  {$x_1$};
                \node[state] (C_1-x_3) [right = 0.5cm of C_1-x_2] {$y_1$};

                \node[state] (C_2-x_2) [below = of C_2] {$y_1$};
                \node[dsqr, label=below:{\small $ (-9,9,0,0,0,0) $}] (C_2-x_1) [left = 0.5cm of C_2-x_2] {$\neg x_1$};
                \node[dbl] (C_2-x_3) [right = 0.5cm of C_2-x_2] {$ z_1 $};

                \node[state] (C_3-x_2) [below = of C_3] {$\neg y_1$};
                \node[dsqr, label=below:{\small $ (0,0,-9,9,0,0) $}] (C_3-x_1) [left = 0.5cm of C_3-x_2] {$\neg x_2$};
                \node[dbl] (C_3-x_3) [right = 0.5cm of C_3-x_2] {$\lnot z_1$};

                \coordinate [below = of C_2-x_2] (C_1-x_3-curve-stop-1);
                
                \coordinate [left = of C_1-x_1] (C_1-x_2-curve-stop-1);
                
                \coordinate [left = of C_2-x_1] (C_2-x_3-curve-stop-1);
                
                \coordinate [below = of C_3-x_3] (C_1-x_3-curve-stop-2) ;
                \coordinate [right = 1.5cm of C_3-x_3] (C_1-x_3-curve-stop-3);

                \coordinate [below right = of C_2-x_3] (C_3-x_3-curve-stop-1);
                \coordinate [below right =of C_2-x_1] (C_3-x_3-curve-stop-2);

		\draw[-{Latex[width=2mm]}] (start) --  (s_init);
		\draw[-{Latex[width=2mm]}] (s_init) --  (C_1);
		\draw[-{Latex[width=2mm]}] (s_init) --  (C_2);
		\draw[-{Latex[width=2mm]}] (s_init) --  (C_3);

                \draw[-{Latex[width=2mm]}] (C_1) -- (C_1-x_1);
                \draw[-{Latex[width=2mm]}] (C_1) -- (C_1-x_2);
                \draw[-{Latex[width=2mm]}] (C_1) -- (C_1-x_3);

                \draw[-{Latex[width=2mm]}] (C_2) -- (C_2-x_1);
                \draw[-{Latex[width=2mm]}] (C_2) -- (C_2-x_2);
                \draw[-{Latex[width=2mm]}] (C_2) -- (C_2-x_3);

                \draw[-{Latex[width=2mm]}] (C_3) --  (C_3-x_1);
                \draw[-{Latex[width=2mm]}] (C_3) -- (C_3-x_2);
                \draw[-{Latex[width=2mm]}] (C_3) --  (C_3-x_3);

		\draw[-{Latex[width=2mm]}] (C_2-x_3) edge[loop, out=290, in=250, distance=0.5cm] node[below]{\small $ (0,0,0,0,9,-9) $} (C_2-x_3);

		\draw[-{Latex[width=2mm]}] (C_3-x_3) edge[loop, out=290, in=250, distance=0.5cm] node[below]{\small $ (0,0,0,0,-9,9) $} (C_3-x_3);

                \draw[-{Latex[width=2mm]}] (C_1-x_1) to [out=90,in=180] (s_init);
                
                \draw[] (C_1-x_2) to [out=-90,in=-90] node[below, xshift=1.5cm]{\small $ (0,0,9,-9,0,0) $} (C_1-x_2-curve-stop-1);
                
                \draw[-{Latex[width=2mm]}] (C_1-x_2-curve-stop-1) to [out=90,in=160] (s_init);
                
                \draw[-{Latex[width=2mm]}] (C_2-x_1) to [out=90,in=-120] (s_init);
                
                \draw[-{Latex[width=2mm]}] (C_3-x_1) to [out=90,in=-45] (s_init);
                
                \draw[] (C_2-x_3) to [out=-120,in=-90] (C_2-x_3-curve-stop-1);
                
                \draw[-{Latex[width=2mm]}] (C_2-x_3-curve-stop-1) to [out=90,in=0] (C_1);

		;
\end{tikzpicture}
}
		\caption{The game arena of $ \Game^\Psi $. White circle states are controlled by $ A $, square by $ E $, grey squares $ x_1, \neg x_1, x_2, \neg x_2 $ by players $ 1, 2, 3, 4 $, respectively, and grey circles $ z_1, \neg z_1 $ by players $ 5, 6 $, respectively. Similar to the illustration in the proof of Theorem~\ref{thm:dominated} (Figure~\ref{fig:qsat-ap}), the weight function is given as vectors, where states without vectors have $ \vec{0} $. We exclude the weight vectors for $ E, A, P, Q, R $ as they are all zeros, to avoid clutters. Moreover, each literal state (i.e., $ x/y/z $ state) also has a transition to $ \sink $. }
		\label{fig:qsat3-ap}
	\end{figure*}
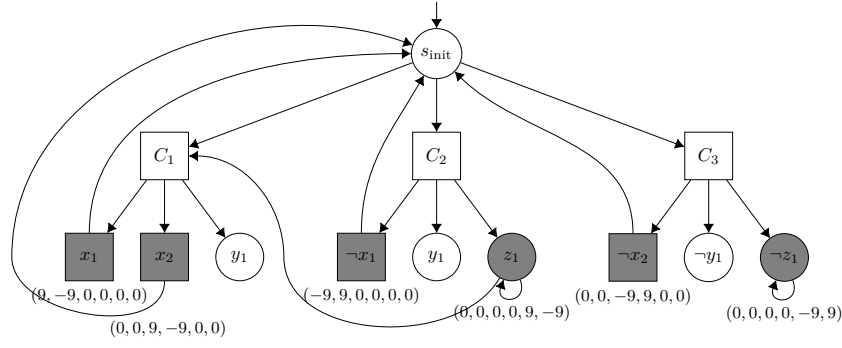
	From this formula, we construct the game $ \Game^\Psi $ in Figure~\ref{fig:qsat3-ap} with $ \Ag = \{1,2,3,4,5,6,E,A,P,Q,R\} $ with weight function given as vectors and assigned to players in an analogous way as in the illustration in the proof of Theorem~\ref{thm:dominated} (Figure~\ref{fig:qsat-ap}). We then ask whether the core of $ \Game^\Psi $ is not empty. Observe that the formula $ \Psi $ is true under the following assignment: $ v(x_1) = v(z_1) = true $ and $ v(x_2) = false $. Indeed by constructing $ \strElm_E $ from it, players $ 1, 4, 6 $ get a payoff of at least 1 for all $ \strElm_A $, and as such there is no beneficial deviation.
	
	The construction above produces a game whose size is polynomial in the size of $ \Psi $. More specifically, the numbers of states and transitions are, respectively, linear and quadratic in the size of $ \Psi $.
	
	We now show that the core of $ \Game^\Psi $ is not empty if and only if $ \Psi $ is satisfiable.
	
	$ (\Rightarrow) $ Suppose $ \strpElm \in \core(\Game^\Psi) $. By the construction of the game, there are three cases:
	\begin{enumerate}
		\item[(a)] $ \pi(\strpElm) $ visits some literal state of the form $ x_k $ (resp. $ \neg x_k $) infinitely often, and $ \pay_{2k-1}(\strpElm) \geq 1 $ (resp. $ \pay_{2k}(\strpElm) \geq 1 $)
		
		\item[(b)] $ \pi(\strpElm) $ visits some literal state of the form $ z_k $ (resp. $ \neg z_k $) infinitely often, and $ \pay_{2(p+k)-1}(\strpElm) \geq 1 $ (resp. $ \pay_{2(p+k)}(\strpElm) \geq 1 $)
		
		\item[(c)] both (a) and (b).
	\end{enumerate}
	The condition $ \pay_i(\strpElm) \geq 1 $ is necessary, because otherwise player $ i $ can deviate to $ \sink $ and gets a payoff of $ 1 $ which contradicts $ \strpElm $ being in the core.
	
	We start with (a). This implies that for each clause $ C_i, 1 \leq i \leq r $, there is a strategy $ \strElm_E $ for player $ E $ that agrees with $ \strpElm $ for choosing a literal state $ l_{i,j} $ such that for a literal of the form $ x_k $ (resp. $ \neg x_k $) we have $ \wFun_{2k-1}(l_{i,j}) \geq 3 $ (resp. $ \wFun_{2k}(l_{i,j}) \geq 3 $). Moreover, if such a strategy exists, then it is a valid assignment for $ x_1,\dots,x_p $ (i.e., contains no contradictions), since otherwise player $ A $ can alternate between the two contradictory choices and gets $ \pay_{2k}(\strpElm) = 0 $ or $ \pay_{2k-1}(\strpElm) = 0 $, which implies that there is a beneficial deviation by player $ 2k $ or $ 2k-1 $--contradicting our assumption that $ \strpElm $ being in the core. Since this assignment is valid and makes all clauses evaluate to true, then it is the case that $ \Psi $ is satisfiable.
	
	For case (b), the argument is similar to (a). The main difference is that from a literal state $ l_{i,j} $ of the form $ z_k $ or $ \neg z_k $, player $ A $ can choose to go to state a $ C' $ that is not \yclash with $ C_i $. This assures that player $ A $ can only choose a valid assignment for $ y_1,\dots,y_q $. Moreover, since we have $ \pay_{2(p+k)-1}(\strpElm) \geq 1 $ or $ \pay_{2(p+k)}(\strpElm) \geq 1 $, then for each clause visited, there exists an assignment of $ z_1,\dots,z_t $ that makes the clause evaluates to true. This assignment is a satisfying assignment for $ \Psi $. For case (c), we combine the arguments from (a) and (b), and obtain a similar conclusion.
	
	$ (\Leftarrow) $ Now, suppose that $ \Psi $ is satisfiable, then we have the following cases:
	\begin{enumerate}
		\item[(1)] there exists an assignment $ v(x_1,\dots,x_p) $ such that $ \Psi(v) $ is a tautology, where $ \Psi(v) $ is the resulting formula after applying the assignment $ v(x_1,\dots,x_p) $. 
		
		\item[(2)] there exists an assignment $ v(x_1,\dots,x_p) $ such that for each assignment $ w(y_1,\dots,y_q) $, there is an assignment $ u(z_1,\dots,z_t) $ that makes $ \Psi(v,w,u) $ evaluates to true.
	\end{enumerate}
	
	For case (1), we start by turning the assignment $ v(x_1,\dots,x_p) $ into a strategy $ \strElm_E $ that prescribes to which $ x $-literal state $ l_{i,j} $ from each clause state $ C_i $ the play must proceed. For instance, if $ v(x_k) $ is true and $ x_k $ is a literal in $ C_i $, then player $ E $ will choose to go to $ x_k $ from $ C_i $. Notice that it may be the case that there are more than one possible ways to choose a literal according to a given assignment, in which we can just arbitrarily choose one. Observe that by following $ \strElm_E $, for all strategy of player $ A $ $ \strElm_A $, corresponding to the assignments of $ y_1,\dots,y_q $, and for all literal state $ x_k $ (resp. $ \neg x_k $) visited infinitely often in $ \pi((\strElm_E,\strElm_A)) $ we have $ \pay_{2k-1}((\strElm_E,\strElm_A)) \geq 1 $ (resp. $ \pay_{2k}((\strElm_E,\strElm_A)) \geq 1 $). This means that $ (\strElm_E,\strElm_A) $ admits no beneficial deviation and thus it is in the core.
	
	For case (2), we perform a similar strategy construction as in (1). First, observe that the resulting formula $ \Psi(v) $ may contain clauses that evaluate to true. We denote this by $ \chi(\Psi(v)) $. Notice that if $ \chi(\Psi(v)) = \{C_v \vert 1 \leq v \leq r \} $, then $ \Psi(v) $ is a tautology---the same as case (1), and we are done. Otherwise, there is $ C_i \notin \chi(\Psi(v)) $ and $ C_i $ contains some $ z $-literals. Now, using $ u(z_1,\dots,z_t) $ we construct a strategy $ \strElm_E' $ that prescribes which $ x $-literal and $ z $-literal to choose from each clause $ C_i $. Since $ \Psi(v,w,u) $ evaluates to true, then for each $ C_i $ it is the case that $ C_i \in \chi(\Psi(v,w,u)) $. This means that for any $ C_i, C_j \notin \chi(\Psi(v)) $ that are visited infinitely often in a play resulting from $ (\strElm_A,\strElm_E') $, there exist no clashing $ z $-literals in $ C_i,C_j $ visited infinitely often. That is, for any $ C_i, C_j \notin \chi(\Psi(v)) $  we have only $ z_k $ (resp. $ \neg z_k $) visited infinitely often, and by the weight function of the game, we have $ \pay_{2(p+k)-1}((\strElm_A,\strElm_E')) \geq 1 $ (resp. $ \pay_{2(p+k)}((\strElm_A,\strElm_E')) \geq 1 $). Thus, it is the case that $ (\strElm_A,\strElm_E') \in \core(\Game^\Psi) $.
	
\end{proof}

\subsection{Proof of Theorem~\ref{thm:gr-ecore}}

\THMgrecore*

\begin{proof}
	
	Recall that a \GRone formula \(\varphi\) has the following form  \[%
	\varphi = \bigwedge_{l = 1}^{m} \always \sometime \psi_{l} \to \bigwedge_{r = 1}^{n} \always \sometime \theta_{r}\text{,}
	\]  and let \(V(\psi_{l})\) and \(V(\theta_r)\) be the subset of states in \(\Game\) that satisfy the Boolean combinations \(\psi_{l}\) and \(\theta_{r}\), respectively.  Observe that property \(\varphi\) is satisfied over a path \(\pi\) if, and only if, either \(\pi\) visits every \(V(\theta_r)\) infinitely many times or visits some of the \(V(\psi_{l})\) only a finite number of times.
	
	For the game \(\Game{[S]}\) and payoff vector $ \vec{x} $, let \(\tuple{V, E, (\wFun_{i}')_{i \in \Ag}}\) be the underlying graph, where \( \wFun_{i}'(v) = \wFun_{i}(s) - x_i \) for every \( i \in \Ag, v \in V \), and \( s \in S \), such that \(v\) corresponds to \(s\). Furthermore, for every edge \(e\in E\), we introduce a variable \(z_e\). The value \(z_e\) is the number of times that the edge \(e\) is used on a cycle. Moreover, let   \(\src(e) = \{v \in V : \exists w\, e = (v,w) \in E\}\);   \(\trg(e) = \{v \in V : \exists w\, e = (w,v) \in E\}\);   \(\OUT(v) = \{e \in E : \src(e) = v\}\);   \(\IN(v) = \{e \in E : \trg(e) = v\}\).
	
	Consider \(\psi_{l}\) for some \(1 \leq l \leq m\), and define the linear program \(\mathcal{L}(\psi_{l})\) with the following inequalities and equations:
	\begin{enumerate}
		\item[Eq1:] \(z_e \geq 0\) for each edge \(e\) --- a basic consistency criterion;
		\item[Eq2:] \(\Sigma_{e \in E} z_e \geq 1\) --- ensures that at least one edge is chosen;
		\item[Eq3:] for each \(i \in \Ag\), \(\Sigma_{e \in E} \wFun_i'(\src(e)) z_e \geq 0\) --- ensures that the total sum of any solution is positive;
		\item[Eq4:] \(\Sigma_{\src(e) \cap V(\psi_{l}) \neq \emptyset} z_e = 0\) --- ensures that no state in \(V(\psi_{l})\) is in the cycle associated with the solution;
		\item[Eq5:] for each \(v \in V\), \(\Sigma_{e \in \OUT(v)} z_e = \Sigma_{e \in \IN(v)} z_e\)  --- says that the number of times one enters a vertex is equal to the number of times one leaves that vertex.
	\end{enumerate}
	
	By construction, it follows that \(\mathcal{L}(\psi_{l})\) admits a solution if and only if there exists a path \(\pi\) in \(\Game{[S]}\) such that \(\pay_i(\pi) \geq x_i \) for every player \(i\) and visits \(V(\psi_{l})\) only \emph{finitely many times}. Furthermore, consider the linear program \(\mathcal{L}(\theta_{1}, \ldots, \theta_{n})\) defined with the following inequalities and equations:
	\begin{enumerate}
		\item[Eq1:] \(z_e \geq 0\) for each edge \(e\) --- a basic consistency criterion;
		\item[Eq2:] \(\Sigma_{e \in E} z_e \geq 1\)  --- ensures that at least one edge is chosen;
		\item[Eq3:] for each \(i \in \Ag\), \(\Sigma_{e \in E} \wFun_i'(\src(e)) z_e \geq 0\) --- ensures that the total sum of any solution is positive;
		\item[Eq4:] for all \(1 \leq r \leq n\), \(\Sigma_{\src(e) \cap V(\theta_{r}) \neq \emptyset} z_e \geq 1\) --- ensures that for every \(V(\theta_{r})\) at least one state is in the cycle;
		\item[Eq5:] for each \(v \in V\), \(\Sigma_{e \in \OUT(v)} z_e = \Sigma_{e \in \IN(v)} z_e\) --- says that the number of times one enters a vertex is equal to the number of times one leaves that vertex.
	\end{enumerate}
	
	In this case, \(\mathcal{L}(\theta_{1}, \ldots, \theta_{n})\)  admits a solution if and only if there exists a path \(\pi\) in \( \Game{[S]} \) such that \(\pay_i(\pi) \geq x_i\) for every player \(i\) and visits every \(V(\theta_{r})\) \emph{infinitely many times}.  Since the constructions above are polynomial in the size of both \(\Game\) and \(\phi\), we can conclude that given \( \Game{[S]} \), vector $ \vec{x} $, and \GRone formula \( \phi \), it is possible to check in polynomial time whether \( \phi \) is satisfied by a suitable path $ \pi $ in \( \Game{[S]} \).
	
	Therefore, to solve \textsc{E-Core} with \GRone specifications, we can use Algorithm~\ref{alg:ecore} with polynomial time check for line 5. Thus, it follows that \textsc{E-Core} with \GRone specifications can be solved in $ \SigmaPThree $. The lower bound follows directly from the hardness result of $ \textsc{Non-Emptiness} $ by setting $ \phi = \top $. Moreover, since $ \acore $ is the dual of $ \ecore $, we obtain the  theorem.
\end{proof}

\end{document}